\documentclass[10pt,twocolumn,twoside]{IEEEtran}
\usepackage{cite}  
\usepackage[pdftex]{graphicx}
\usepackage{stfloats} 
\usepackage{amsmath,amsfonts,amssymb,amscd}
\usepackage{multirow}
\usepackage{bbm}
\usepackage{bbold}

\usepackage{amsthm}
\newtheorem{proposition}{Proposition}

\newtheorem{lemma}{Lemma}

\newtheorem{assumption}{Assumption}

\usepackage{setspace}
\usepackage{pdfsync}
\usepackage{epstopdf}
\usepackage{hyperref}
\usepackage{cleveref}
\usepackage{lipsum,multicol}
\usepackage[numbered,autolinebreaks,useliterate]{mcode}
\usepackage{mathtools,bbm}
\usepackage{mathtools, dsfont}

\usepackage{caption}
\usepackage{subcaption}
\usepackage{scalefnt}

\usepackage{lipsum}
\usepackage{calligra}
\DeclareMathAlphabet{\mathcalligra}{T1}{calligra}{m}{n}

\DeclareMathOperator{\diag}{diag}

\newcommand{\vertiii}[1]{{\left\vert\kern-0.25ex\left\vert\kern-0.25ex\left\vert #1 
    \right\vert\kern-0.25ex\right\vert\kern-0.25ex\right\vert}}

\usepackage[noend]{algpseudocode}
\usepackage{algorithm,algorithmicx}

\algrenewcommand\algorithmicrequire{\textbf{Initialization:}}
\algrenewcommand\algorithmicensure{\textbf{For each $i$, execute for $k\geq 0$:}}
\algblockdefx[Loop]{Loop}{EndLoop}[1]{\textbf{for $k\geq 0$:} #1}

\makeatletter
\algrenewcommand\ALG@beginalgorithmic{\small}
\makeatother

\makeatletter
\newcommand\fs@betterruled{%
  \def\@fs@cfont{\bfseries}\let\@fs@capt\floatc@ruled
  \def\@fs@pre{\vspace*{5pt}\hrule height.8pt depth0pt \kern2pt}%
  \def\@fs@post{\kern2pt\hrule\relax}%
  \def\@fs@mid{\kern2pt\hrule\kern2pt}%
  \let\@fs@iftopcapt\iftrue}
\floatstyle{betterruled}
\restylefloat{algorithm}
\makeatother

\usepackage{balance}
\usepackage{float}
\usepackage{tabu}

\allowdisplaybreaks
\begin{document}
\title{{Distributed Optimal Power Flow Algorithms over Time-Varying Communication Networks}}

%\author{Madi Zholbaryssov and Alejandro D. Dom\'inguez-Garc\'ia
%\thanks{The authors are with the Department of Electrical and Computer Engineering of the University of Illinois at Urbana-Champaign, Urbana, IL 61801,
%USA. Email: {\tt\small \{zholbar1, aledan\}@ILLINOIS.EDU.}
%}}
\author{Madi Zholbaryssov, Alejandro D. Dom\'{i}nguez-Garc\'{i}a,~\IEEEmembership{Senior Member,~IEEE}
\thanks{M. Zholbaryssov and A. D. Dom\'{i}nguez-Garc\'{i}a are with the ECE Department at the University of Illinois at Urbana-Champaign, Urbana, IL 61801, USA. E-mail:  \{zholbar1, aledan\}@ILLINOIS.EDU.}}
\maketitle
\begin{abstract}
In this paper, we consider the problem of optimally coordinating the response of a group of distributed energy resources (DERs) in distribution systems by solving the so-called optimal power flow (OPF) problem. The OPF problem is concerned with determining an optimal operating point, at which some cost function, e.g., generation cost or power losses, is minimized, and operational constraints are satisfied. To solve the OPF problem, we propose distributed algorithms that are able to operate over time-varying communication networks and have geometric convergence rate. We solve the second-order cone program
(SOCP) relaxation of the OPF problem for radial distribution systems, which is formulated using the so-called DistFlow model.
Theoretical results are further supported by the numerical simulations.
\end{abstract}
%\begin{IEEEkeywords}
%Microgrids, ancillary services, distributed control, algorithm.
%\end{IEEEkeywords}
\IEEEpeerreviewmaketitle

\section{Introduction}
%Over the last decade, a significant amount of renewable generation has been integrated into different power grids across the globe. 
%If this integration continues at such a fast rate, power grids might gradually transition from centralized power generation towards predominantly distributed generation with large-scale penetration of distributed energy resources that produce electricity close to the end users of power.
%It is also envisioned that such future power grids will have plug-and-play functionality, which will allow more electricity consumers to seamlessly connect their solar or wind power stations.
%Motivated by a future vision that present-day power grids mainly dependent on centralized power generation stations will transition towards more distributed power generation mostly based on DERs
The growing integration of renewable generation inevitably raises concerns about the efficacy of the traditional centralized electricity paradigm in meeting global energy needs in a reliable and sustainable way. It is envisioned that present-day power grids, which are predominantly based on centralized energy sources, will eventually transition towards more decentralized architecture, where electric power is mainly produced from distributed energy resources (DERs). One of the obstacles in fulfilling this vision is the need for effective control strategies for coordinating a large population of DERs. 

In order to adequately respond to rapid and large fluctuations in renewable generation, DERs will need to more frequently adjust their set-points, which will require real-time control systems to run and process data more often. To coordinate a large number of DERs, it will also be required to process a large volume of data in real-time. The traditional centralized approach operates by collecting the data in a central processing unit requiring a denser communication network with high-speed communication channels, which also needs to be secure to prevent cyber attackers from stealing sensitive private information. Hence, the centralized approach may not be feasible because of the resulting communication overhead.  
%The traditional centralized approach, which requires this data to be collected in a central processing unit, may not be feasible because of the resulting communication overhead.  
%which entails development of fast control strategies. 
%Also, because of the communication overhead, it may not be feasible to use a centralized approach to coordinate a large number of DERs over a large geographic area. This necessitates development of distributed control strategies for DER coordination that scale well to power systems of large size.
%To enable high penetration of distributed energy resources, both computational and communicational challenges need to be resolved. To coordinate a large population of DERs, it will be required to process a large amount of data in real-time. If centralized approach is used, this data needs to be collected in a single node, which may not be feasible because of the communication overhead and constraints. 
%Although it is significantly faster to solve the OPF problem in a centralized way, collecting real-time data in a central processing unit will require building a denser communication network with high-speed communication channels, which also needs to be secure to prevent cyber attackers from stealing sensitive private information.
By contrast, a distributed approach processes the data locally, thereby dispensing with the need for moving the data to the central node. However, it is more difficult to solve the OPF problem in a distributed way, since communication latency and random data packet losses might prevent the distributed approach from converging to an optimal solution.
%it is computationally challenging to coordinate a large population of DERs.

In this work, we consider the standard OPF problem for balanced distribution systems with high penetration of DERs, where each DER is operated within its capacity constraints and is endowed with a generation cost. The objective of the OPF problem is to determine an optimal operating point at which some cost function, e.g., total generation cost or power loss, is minimized, and operational constraints are satisfied. 
%At each bus, there is a DER that can only be operated within its capacity constraints and has generation cost.
%When formulating the OPF problem, we adopt the DistFlow model \cite{BaWu89}
%We assume that a computing device is attached to each bus in the distribution system, and is able to communicate with the computing devices at neighboring buses.
%Then, the objective is to determine optimal power outputs of the DERs in a distributed manner to minimizing the total generation cost and respecting the operational constraints. 
One of the main objectives of this work is to design distributed OPF solvers with a certain level of resiliency to communication latency and random data packet losses required to properly coordinate DERs over less reliable communication networks. Another important objective is to ensure that these algorithms have fast convergence rate in order to quickly update the set-points of DERs and provide a fast response to rapid and large transients in renewable generation.

A vast body of work has focused on solving the OPF problem for distribution systems. Earlier works (see, e.g., \cite{BaWe08, FaLo13, GaLiLo15}) focused on dealing with the non-convexity of the OPF problem, and proposed semidefinite program (SDP) and SOCP relaxations, which were shown to be exact for radial networks under some conditions. A few works proposed distributed approaches for solving the OPF problem over time-invariant communication networks (see, e.g., \cite{DaZh13,SuBa14,MaWe15,PeLo18}). In \cite{PeLo18}, the authors proposed a distributed algorithm for solving the SOCP relaxation of the OPF problem for balanced radial networks, which is based on the alternating direction method of multipliers (ADMM). The asynchronous ADMM was proposed in \cite{GuHu17} to solve the OPF problem over time-varying communication networks.
There exists another body of works (see, e.g., \cite{ZhCh12, DoCaHa12,YaTa13,KaHu12, ZhPa14, CaDoHa15, ChCo15, WuJo17, MoHu14, MoHu15, ZhFo19}) that focused on solving simplified OPF problems over time-varying communication networks, in which most operational constraints were neglected. One line of works (see, e.g., \cite{ZhCh12, DoCaHa12,YaTa13,KaHu12, ZhPa14, CaDoHa15, ChCo15, WuJo17}) focused on the DER coordination problem with only total active power balance constraint and generation capacity constraints being considered. In addition to these constraints, another line of works (see, e.g., \cite{MoHu14, MoHu15, ZhFo19}) also considered line flow constraints. 
%However, we note that although these works were proposed to solve a simplified OPF problem, they were valuable in that they allow us to understand 

%The authors of \cite{Nedic17} propose distributed algorithms with geometric convergence rate for solving an unconstrained optimization problem over time-varying graphs. However, since the DER coordination problem is a constrained optimization problem, these algorithms cannot be directly applied to solve it. 
Our starting point in the design of the algorithms is a primal-dual algorithm for solving the system of optimality conditions also known as the \textit{Lagrangian system}. We then develop distributed versions of this primal-dual algorithm by having bus agents closely emulate the iterations of the primal-dual algorithm, where each agent maintains and updates only local variables.
%To handle the power balance constraint, we introduce a dual variable and update them using the algorithms in \cite{Nedic17}. These dual variables are further used to update the primal variables. 
The resulting distributed primal-dual algorithms converge geometrically fast when operated over unreliable communication networks. Each algorithm can be viewed as a feedback interconnection of the primal-dual algorithm representing the nominal system and the disturbance generated due to the nature of the distributed implementation. The key ingredient for establishing the convergence results is to show that both systems are finite-gain stable, which then allows us to use the small-gain theorem (see, e.g., \cite{Khalil}) and show the convergence of the feedback interconnected system. The small-gain-theorem-based analysis first appeared in \cite{Nedic17} in the context of distributed algorithms for solving an unconstrained consensus optimization problem.
%However, since the OPF problem is a constrained optimization problem, these algorithms cannot be directly applied to solve it.

%The rest of the paper is organized as follows. In Section~\ref{sec:preliminaries}, we formulate the OPF problem and its convex relaxation. In Section~\ref{sec:opf}, we present a distributed primal-dual algorithm over time-varying communication networks for radial distribution systems. In Section~\ref{sec:PD2}, we present another distributed primal-dual algorithm over time-varying communication networks for mesh distribution systems. 
%%In Section~\ref{sec:rPD}, we present robust distributed primal-dual algorithm. 
%In Section~\ref{sec:simulations}, we demonstrate the performance of the algorithms using the IEEE~$33$-- and $69$--bus radial test systems and the IEEE~$118$--bus mesh test system. Concluding remarks are in Section~\ref{sec:conclusion}. 

\section{Preliminaries}
\label{sec:preliminaries}
In this section, we formulate the OPF problem and outline the communication network models adopted in this work.
\subsection{OPF Problem Formulation}
We consider a balanced radial distribution system, the topology of which can be described by a directed graph~$\mathcal{G}_p = (\mathcal{V}_p,\mathcal{E}_p)$, where $\mathcal{V}_p \coloneqq \{1,2,\dots, n\}$ denotes the set of $n$ buses (nodes), and $\mathcal{E}_p \coloneqq \{\vec{e}_1, \vec{e}_2, \dots, \vec{e}_{n-1}\}$ denotes the set of electrical lines, with $\vec{e}_k = (i,j) \in \mathcal{E}_p$ if node~$i$ is located upstream from node~$j$, i.e., node~$i$ is closer to the distribution substation, represented by node~$1$. Let $\mathcal{N}_i^+ = \{j: (i,j) \in \mathcal{E}_p\}$ and $\mathcal{N}_i^- = \{l: (l,i) \in \mathcal{E}_p\}$ denote the sets of (downstream) out-neighbors and (upstream) in-neighbors of node~$i$, respectively. We define a node-to-edge incidence matrix, $M\in\mathds{R}^{n\times |\mathcal{E}_p|}$, with $M_{ik} = 1$ and $M_{jk} = -1$, if $\vec{e}_k = (i,j) \in \mathcal{E}_p$, and $M_{ik} = 0$ and $M_{jk} = 0$, otherwise. Also, let $M_0\in\mathds{R}^{n\times |\mathcal{E}_p|}$ contain the entries of $M$, each corresponding to an upstream end, with $M_{0ik} = 1$ if $\vec{e}_k = (i,j) \in \mathcal{E}_p$, for some $j\in\mathcal{V}_p$, and $M_{0ik} = 0$, otherwise. Similarly, $N_0\in\mathds{R}^{n\times |\mathcal{E}_p|}$ contains the entries of $M$, each corresponding to a downstream end, with $N_{0ik} = -1$ if $\vec{e}_k = (j,i) \in \mathcal{E}_p$, for some $j\in\mathcal{V}_p$, and $N_{0ik} = 0$, otherwise. [Note that $M = M_0+N_0$.] Let $r_{ij}$ and $x_{ij}$ denote the series resistance and reactance of line~$(i,j)$, $R \coloneqq \diag(\{r_{ij}\}_{(i,j)\in\mathcal{E}_p})$, and $X \coloneqq \diag(\{x_{ij}\}_{(i,j)\in\mathcal{E}_p})$.

Let $l^{(p)}_i$ and $g^{(p)}_i$ denote the active power demand and supply at node~$i$, $l^{(q)}_i$ and $g^{(q)}_i$ denote the reactive power demand and supply at node~$i$, $p_{ij}$ and $q_{ij}$ denote the active and reactive power flow out of node~$i$ through line~$(i,j)$, i.e., $p_{ij} + q_{ij}\sqrt{-1}$ is the sending end power. Let $V_i$ denote the voltage magnitude at node~$i$, and $v_i \coloneqq V_i^2$. Let $I_{ij}$ denote the current magnitude through line~$(i,j) \in \mathcal{E}_p$, and $\ell_{ij} \coloneqq I_{ij}^2$. If $\mathcal{G}_p$ is radial, the AC power flow equations can be exactly represented via the standard \textit{DistFlow} model (see, e.g., \cite{BaWu89,FaLo13}):
\begin{subequations}\label{DistFlow}
\begin{align}
0 &= g^{(p)} - l^{(p)} - Mp + N_0R\ell,\label{DistFlow1}\\
0 &= g^{(q)} - l^{(q)} - Mq + N_0X\ell,\label{DistFlow2}\\
0 &= M^\top v - 2Rp-2Xq+(R^2+X^2)\ell,\label{DistFlow3}\\
0 &= p\circ p+q\circ q - M_0^\top v\circ \ell,\label{DistFlow4}
\end{align}
\end{subequations}
where $\circ$ denotes an element-wise multiplication, $g^{(p)} = [g^{(p)}_1,\dots,g^{(p)}_n]^\top$, $g^{(q)} = [g^{(q)}_1,\dots,g^{(q)}_n]^\top$, $v = [v_1,\dots,v_n]^\top$, $p = [\{p_{ij}\}_{(i,j)\in\mathcal{E}_p}]$, $q = [\{q_{ij}\}_{(i,j)\in\mathcal{E}_p}]$, and $\ell = [\{\ell_{ij}\}_{(i,j)\in\mathcal{E}_p}]$.

Next, we specify the operational constraints to be satisfied in the problem formulation. We impose the following operating limits on the power outputs of the DERs: 
\begin{subequations}\label{cap_constr}
\begin{align}
g^{(p)\min}&\leq g^{(p)}\leq g^{(p)\max},\\
g^{(q)\min}&\leq g^{(q)}\leq g^{(q)\max}.
\end{align}
\end{subequations}
%where $g^{(w)\min} \coloneqq [g^{(w)\min}_1,\dots,g^{(w)\min}_n]^\top$, $g^{(w)} \coloneqq [g^{(w)}_1,\dots,g^{(w)}_n]^\top$, $g^{(w)\max} \coloneqq [g^{(w)\max}_1,\dots,g^{(w)\max}_n]^\top$, $w\in\{p,q\}$.
For an inverter-interfaced DER at node~$i$, the maximum amount of active and reactive power that can be produced is limited by the apparent power capability, $s^{\max}_i$, of an inverter:
\begin{align}
(g^{(p)}_i)^2 + (g^{(q)}_i)^2 \leq (s^{\max}_i)^2. \label{cap_constr1}
\end{align}
To simplify the exposition, we will not consider \eqref{cap_constr1} in the problem formulation; however, the proposed algorithms can handle the constraint \eqref{cap_constr1}.

Also, the voltage levels and line currents need to be within the following operating limits:
\begin{align}
v^{\min}\leq v\leq v^{\max}, \label{volt_constr}\\
0 \leq \ell \leq \ell^{\max}. \label{curr_constr}
\end{align}
%where $v \coloneqq [v_1,\dots,v_n]^\top$, $v^{\min} \coloneqq [v^{\min}_1,\dots,v^{\min}_n]^\top$, $v^{\max} \coloneqq [v^{\max}_1,\dots,v^{\max}_n]^\top$, $\ell^{\max} \coloneqq [\ell^{\max}_1,\dots,\ell^{\max}_n]^\top$.
Then, the OPF problem can be formulated as follows:
\begin{align}
{\textbf{OPF}}:
\mbox{min} & \mbox{ } f(g^{(p)})\coloneqq \sum\limits_{i=1}^n f_i(g^{(p)}_i) \nonumber\\
\mbox{over} & \mbox{ } g^{(p)}, g^{(q)}, p, q, v, \ell\label{OPF}\\
\mbox{subject to} & \mbox{ \eqref{DistFlow}, \eqref{cap_constr}, \eqref{volt_constr}, \eqref{curr_constr}},\nonumber
%&\mbox{ }g^{(p)\min}\leq g^{(p)}\leq g^{(p)\max},\label{gen_constr1}\\
%&\mbox{ }\underline{g}^q\leq g^{(q)}\leq \overline{g}^q,\label{gen_constr2}\\
%&\mbox{ }\underline{V}\leq V\leq \overline{V},\label{volt_constr}
\end{align}
where $f_i(\cdot)$ denotes the cost function associated with the electric power generated by the DER at bus~$i$. We make the following assumption regarding the cost function. 
\begin{assumption}\label{objective_assumption}
Each cost function~$f_i(\cdot)$ is twice differentiable and strongly convex with parameter~$m>0$, i.e., $f_i^{\prime\prime}(x)\geq m$, $\forall x \in [g^{(p)\min}_i, g^{(p)\max}_i]$, $\forall i\in\mathcal{V}_p$.
\end{assumption}

\subsection{SOCP Relaxation Of The OPF Problem}
Because of the nonlinear equality constraint \eqref{DistFlow4}, the OPF problem \eqref{OPF} is non-convex. For radial networks, it has been shown in \cite{FaLo13,GaLiLo15} that under certain assumptions when \eqref{DistFlow4} is relaxed to the second-order cone constraint, namely,
\begin{align}\label{DistFlow5}
p\circ p+q\circ q \leq M_0^\top v\circ\ell,
%v_i\ell_{ji} &\geq p_{ji}^2 + q_{ji}^2, \mbox{ }(j,i) \in \mathcal{E}_p,
\end{align}
the OPF problem \eqref{OPF} admits an exact second-order cone program (SOCP) relaxation given below:
\begin{align}
{\textbf{SOCP}}:
\mbox{minimize} & \mbox{ } f(g^{(p)}) \nonumber\\
\mbox{over} & \mbox{ } g^{(p)}, g^{(q)}, p, q, v, \ell\label{OPF2}\\
\mbox{subject to} & \mbox{ \eqref{DistFlow1}--\eqref{DistFlow3}, \eqref{cap_constr}, \eqref{volt_constr}, \eqref{curr_constr}, \eqref{DistFlow5}}.\nonumber
%&\mbox{ }g^{(p)\min}\leq g^{(p)}\leq g^{(p)\max},\label{gen_constr1}\\
%&\mbox{ }\underline{g}^q\leq g^{(q)}\leq \overline{g}^q,\label{gen_constr2}\\
%&\mbox{ }\underline{V}\leq V\leq \overline{V},\label{volt_constr}
\end{align}
For our further analysis, we introduce additional variable~$\varepsilon = [\{\varepsilon_{ij}\}_{(i,j)\in\mathcal{E}_p}]$, and break the constraint \eqref{DistFlow3} into the following equivalent constraints:
\begin{subequations}
\begin{align}
0 &= \varepsilon - 2Rp-2Xq+(R^2+X^2)\ell,\\
0 &= \varepsilon - M^\top v.
%0 &= \varepsilon_{ji} - 2(r_{ji}p_{ji}+x_{ji}q_{ji}) + (r_{ji}^2+x_{ji}^2)\ell_{ji},\\
%\varepsilon_{ji} &= v_j - v_i, \mbox{ }j\in\mathcal{U}_i.
\end{align}
\label{DistFlow6}%
\end{subequations}
The proposed algorithms are designed to solve a regularized approximation of the OPF problem \eqref{OPF2}, where we add a regularization term to the objective function. The purpose of the problem regularization is to allow us to establish the convergence results.
%The proposed algorithm relies on the use of the regularization term that plays an important role in establishing the convergence results. 
However, if the regularization term is small, there is practically no difference between the solutions of \eqref{OPF2} and its regularized approximation, which we provide below:
\begin{align}
{\textbf{rSOCP}}:
\mbox{minimize} & \mbox{ } f(g^{(p)}) + \rho \|\ell\|_2^2 + \rho\|v-1\|_2^2\nonumber\\
\mbox{over} & \mbox{ } g^{(p)}, g^{(q)}, p, q, v, \ell, \varepsilon \label{rOPF}\\
\mbox{subject to} & \mbox{ \eqref{DistFlow1}--\eqref{DistFlow2}, \eqref{cap_constr}, \eqref{volt_constr}, \eqref{curr_constr}, \eqref{DistFlow5}, \eqref{DistFlow6}},\nonumber
%&\mbox{ }g^{(p)\min}\leq g^{(p)}\leq g^{(p)\max},\label{gen_constr1}\\
%&\mbox{ }\underline{g}^q\leq g^{(q)}\leq \overline{g}^q,\label{gen_constr2}\\
%&\mbox{ }\underline{V}\leq V\leq \overline{V},\label{volt_constr}
\end{align}
where $\|\cdot\|_2$ is the Euclidean norm, and $\rho \|\ell\|_2^2+ \rho\|v-1\|_2^2$ is the regularization term that also allows us to penalize the line currents and the deviation of the bus voltages from the nominal voltage, $1$~pu.
To this end, we develop a distributed algorithm that solves rSOCP for radial distribution systems. 

\subsection{Communication Network Models}\label{subsec:cyber_layer}
Next, we introduce the model describing the communication network that enables the information exchange between nodes of the distribution system. We assume that the topology of the nominal communication network coincides with the topology of the power network. We consider (i) bidirectional and (ii) unidirectional communication models.

\subsubsection{Bidirectional Communication Model}
Let $\mathcal{G}_0 = (\mathcal{V}_p,\mathcal{E}_0)$ denote an undirected graph, where $\mathcal{E}_0$ is the set of all available bidirectional communication links, with $\{i,j\} \in \mathcal{E}_0$ if $(i,j)$ or $(j,i) \in \mathcal{E}_p$.
%i.e., $\{i,j\} \in \mathcal{E}_0$ if communication link between $i$ and $j$ is established infinitely often.
During time period~$(t_k,t_{k+1})$, successful data transmissions among nodes can be captured by the undirected graph~$\mathcal{G}^{(c)}[k]=(\mathcal{V}_p,\mathcal{E}_c[k])$, where $\mathcal{E}_c[k]\subseteq \mathcal{E}_0$ is the set of active communication links, with $\{i,j\} \in \mathcal{E}_c[k]$ if nodes~$i$ and $j$ receive information from each other during time period~$(t_k,t_{k+1})$. 
%Let $\mathcal{N}_i[k]$ denote the set of neighbors of node $i$ during time period $(t_k,t_{k+1})$, i.e., $\mathcal{N}_i[k]\coloneqq\{j\in\mathcal{V}_p:\{i,j\} \in \mathcal{E}_c[k]\}$.
We make the following standard assumption regarding the connectivity of the network (see, e.g., \cite{Nedic09}). 
\begin{assumption}
\label{assume_comm_model_un}
$\{i,j\} \in \bigcup_{l= kB}^{(k+1)B-1}\mathcal{E}_c[l]$, $\forall (i,j)\in\mathcal{E}_p$, for some positive integer~$B$.
%There exists some positive integer $B$ such that the graph $\bigcup_{l= kB}^{(k+1)B-1}\mathcal{E}_c[l]$ is connected for $k=0,1,\dots$. 
\end{assumption}
Assumption~\ref{assume_comm_model_un} requires that a communication link~$\{i,j\}$ is active at least once every $B$ iterations. Note that communication graph~$\mathcal{G}^{(c)}[k]$ is not necessarily connected at any given time instant~$k$. 

\subsubsection{Unidirectional Communication Model}
Let $\vec{\mathcal{G}}_0 = (\mathcal{V}_p,\vec{\mathcal{E}}_0)$ denote a directed graph, where $\vec{\mathcal{E}}_0$ is the set of all available unidirectional communication links, with $(i,j) \in \vec{\mathcal{E}}_0$ and $(j,i) \in \vec{\mathcal{E}}_0$ if $(i,j) \in \mathcal{E}_p$. 
%i.e., $\{i,j\} \in \mathcal{E}_0$ if communication link between $i$ and $j$ is established infinitely often.
During time period~$(t_k,t_{k+1})$, successful data transmissions among nodes can be captured by the directed graph~$\vec{\mathcal{G}}^{(c)}[k]=(\mathcal{V}_p,\vec{\mathcal{E}}_c[k])$, where $\vec{\mathcal{E}}_c[k]\subseteq \vec{\mathcal{E}}_0$ is the set of active communication links, with $(i,j) \in \vec{\mathcal{E}}_c[k]$ if node~$j$ receives information from node~$i$ during time period~$(t_k,t_{k+1})$. 
%Let $\mathcal{N}_i[k]$ denote the set of neighbors of node $i$ during time period $(t_k,t_{k+1})$, i.e., $\mathcal{N}_i[k]\coloneqq\{j\in\mathcal{V}_p:\{i,j\} \in \mathcal{E}_c[k]\}$.
We make the following assumption regarding the connectivity of the network.
\begin{assumption}
\label{assume_comm_model_dd}
$(i,j) \in \bigcup_{l= kB}^{(k+1)B-1}\vec{\mathcal{E}}_c[l]$ and $(j,i) \in \bigcup_{l= kB}^{(k+1)B-1}\vec{\mathcal{E}}_c[l]$, $\forall (i,j)\in\mathcal{E}_p$, for some positive integer~$B$.
%There exists some positive integer $B$ such that the graph $\bigcup_{l= kB}^{(k+1)B-1}\mathcal{E}_c[l]$ is connected for $k=0,1,\dots$. 
\end{assumption}
%Assumption~\ref{bidir_comm_model} only requires a communication graph to be connected over a longer finite time interval rather than at every time instant.
Unlike the bidirectional communication network, the unidirectional communication network does not necessarily allow nodes to exchange information simultaneously, within the same time period~$(t_k,t_{k+1})$, if a communication link between them is active. 

\subsection{Simplified OPF}
Many of the basic ideas behind the proposed distributed algorithms can be more effectively presented by considering a simplified OPF problem, referred to as the security-constrained economic dispatch (SCED) problem given below:
\begin{subequations}\label{SCED}
\begin{align}
{\textbf{SCED}}:\mbox{minimize} & \mbox{ } f(g^{(p)}) \\
\mbox{over} & \mbox{ } g^{(p)}, p\\
\mbox{subject to} & \mbox{ }g^{(p)}-l^{(p)} = Mp,\label{sced_power_balance}\\
& \mbox{ } g^{(p)\min}\leq g^{(p)}\leq g^{(p)\max},\\
& \mbox{ } p^{\min}\leq p\leq p^{\max},
\end{align}
\end{subequations}
where $p^{\max}$ denotes the vector of the line capacities, and $p^{\min}=-p^{\max}$.
Then, the basic ideas can be directly carried over to design distributed algorithms that solve rSOCP. In the remainder, we first present the distributed primal-dual algorithms that solve SCED. Then, we extend these algorithms to solve rSOCP.

\subsection{The Small-Gain Theorem}
In the following, we give a brief overview of the main analysis tool used in later developments---the small-gain theorem (see, e.g., \cite[Theorem~5.6]{Khalil}) for discrete-time systems.
For the forthcoming developments, we adopt the appropriate metric for measuring energy content of the signals of interest. For a given sequence of iterates, $\{x[k]\}_{k=0}^{\infty}$, where $x[k]\in\mathds{R}^n$, consider the following norm (previously used in \cite{Nedic17}): \[\|x\|_2^{a,K}\coloneqq \max\limits_{0\leq k\leq K}a^{-k}\|x[k]\|_2,\] for some $a \in (0,1)$, where $\|\cdot\|_2$ is the Euclidean norm. If $\|x\|_2^{a,K}$ is bounded for all $K\geq 0$, then, $a^{-k}\|x[k]\|_2$ is bounded for all $k\geq 0$, and, thus, it follows that $x[k]$ converges to zero at a geometric rate~$\mathcal{O}(a^k)$.

Now, consider a feedback connection of two discrete-time systems~$\mathcal{H}_1$ and $\mathcal{H}_2$ such that
%\begin{subequations}\label{feedback-connection}
\begin{align*}
e_2[k+1] &= \mathcal{H}_1(e_1[k]),\\
e_1[k+1] &= \mathcal{H}_2(e_2[k]).
\end{align*}
%\end{subequations}
We assume that $\mathcal{H}_1$ and $\mathcal{H}_2$ are finite-gain stable in the sense of the norm~$\|\cdot\|_2^{a,K}$, namely, the following relations hold:
\begin{subequations}\label{finite-gain-stability}
\begin{align}
\|e_2\|_2^{a,K} &\leq \gamma_1\|e_1\|_2^{a,K}+\beta_1,\\
\|e_1\|_2^{a,K} &\leq \gamma_2\|e_2\|_2^{a,K}+\beta_2,
\end{align}
\end{subequations}
for some nonnegative constants~$\beta_1$, $\beta_2$, $\gamma_1$, and $\gamma_2$.
From~\eqref{finite-gain-stability}, we have that
\begin{align}\label{small-gain-th-1}
\|e_2\|_2^{a,K} &\leq \gamma_1\|e_1\|_2^{a,K}+\beta_1\nonumber\\
&\leq \gamma_1\gamma_2\|e_2\|_2^{a,K}+\gamma_1\beta_2+\beta_1,
\end{align}
which by rearranging~\eqref{small-gain-th-1} yields
\begin{align*}%\label{small-gain-th-2}
\|e_2\|_2^{a,K} &\leq \frac{\gamma_1\beta_2+\beta_1}{1-\gamma_1\gamma_2}.
\end{align*}
Similarly,
\begin{align*}%\label{small-gain-th-3}
\|e_1\|_2^{a,K} &\leq \frac{\gamma_2\beta_1+\beta_2}{1-\gamma_1\gamma_2}.
\end{align*}
Then, if $\gamma_1\gamma_2<1$, $\|e_1\|_2^{a,K}$ and $\|e_2\|_2^{a,K}$ are bounded, and $e_1[k]$ and $e_2[k]$ converge to zero at a geometric rate~$\mathcal{O}(a^k)$. 
%This type of result is referred to as the small-gain theorem \cite[Theorem~5.6]{Khalil}.

\section{Distributed SCED Over Time-Varying Communication Graphs}\label{sec:sced}
The purpose of this section is to present the key ideas behind the distributed primal-dual algorithm for solving rSOCP by considering a simpler problem, namely, SCED. We first consider the case of undirected communication graphs. Then, we tackle the general case of directed communication graphs.
\subsection{Time-Varying Undirected Communication Graphs}\label{subsec:sced_bi}
%Later In Section~\ref{sec:opf}, we expand on these ideas and develop the distributed primal-dual algorithm for solving rSOCP.
Let $L(g^{(p)},p,\lambda)$ denote the Lagrangian for SCED given by
\begin{align*}
L(g^{(p)},p,\lambda)  &= f(g^{(p)}) + \lambda^\top (g^{(p)} - l^{(p)} - Mp) \\&\quad+\frac{\rho}{2}\|g^{(p)}- l^{(p)} - Mp\|_2^2,
\end{align*}
where $\rho>0$ is a regularization parameter, $\lambda$ denotes the Lagrange multiplier associated with the power balance constraint \eqref{sced_power_balance}.
Our starting point for solving SCED is the following primal-dual algorithm \cite[Section~4.4]{NonlinearProgramming} with additional projection:
\begin{subequations}\label{centr_pd_sced}
\begin{align}
g^{(p)}[k+1] &= \left[g^{(p)}[k] - s\frac{\partial L[k]}{\partial g^{(p)}}\right]_{g^{(p)\min}}^{g^{(p)\max}},\label{pd_sced1}\\
p[k+1] &= \left[p[k] - s\frac{\partial L[k]}{\partial p}\right]_{p^{\min}}^{p^{\max}},\label{pd_sced2}\\
\lambda[k+1] &= \lambda[k] + s\frac{\partial L[k]}{\partial \lambda},
\end{align}
\end{subequations}
where $s>0$ is a stepsize, and $[\cdot]_{x_1}^{x_2}$ denotes the projection onto the box~$[x_1,x_2]$, for $x_1, x_2 \in\mathds{R}^n$. Let $(g^{(p)*},p^*,\lambda^*)$ denote the equilibrium of~\eqref{centr_pd_sced}, where $g^{(p)*} \coloneqq [g^{(p)*}_1,\dots,g^{(p)*}_n]^\top$, $p^* \coloneqq [\{p_{ij}^*\}_{(i,j)\in\mathcal{E}_p}]$, and $\lambda^*\coloneqq[\lambda^*_1,\dots,\lambda^*_n]^\top$. In the distributed version of \eqref{centr_pd_sced}, each node~$i$ maintains the estimates of only local optimal quantities, namely, the local optimal power injection, $g^{(p)*}_i$, out-going flows, $p_{ij}^*$, $(i,j)\in\mathcal{E}_p$, in-coming flows, $p_{li}^*$, $(l,i)\in\mathcal{E}_p$, and the Lagrange multiplier, $\lambda_i^*$. Let $g^{(p)}_i[k]$ denote the estimate of $g^{(p)*}_i$, $\hat{p}_{ij}[k]$ the estimate of $p_{ij}^*$, $(i,j)\in\mathcal{E}_p$, $\breve{p}_{li}[k]$ the estimate of $p_{li}^*$, $(l,i)\in\mathcal{E}_p$, and $\lambda_i[k]$ the estimate of $\lambda_i^*$ maintained at node~$i$ at time instant~$k$. Note that $p_{ij}^*$, $(i,j)\in\mathcal{E}_p$, is estimated by nodes~$i$ and $j$. Let \[x^{(i)}[k] \coloneqq \big[g^{(p)}_i[k],[\{\hat{p}_{ij}[k]\}_{(i,j)\in\mathcal{E}_p}],[\{\breve{p}_{li}[k]\}_{(l,i)\in\mathcal{E}_p}],\lambda_i[k]\big]^\top\] denote the vector of the estimates of all local optimal primal and dual variables, denoted by 
\[x^{(i)*} \coloneqq \big[g^{(p)*}_i,[\{\hat{p}_{ij}^*\}_{(i,j)\in\mathcal{E}_p}],[\{\breve{p}_{li}^*\}_{(l,i)\in\mathcal{E}_p}],\lambda_i^*\big]^\top,\]
maintained at node~$i$ at time instant~$k$.
We let node~$i$ perform the updates based on its local Lagrangian given by
\begin{align*}
L^{(i)}(x^{(i)})  &= f(g^{(p)}_i) + \lambda_ib^{(p)}_i + \frac{\rho}{2}(b^{(p)}_i)^2,
\end{align*}
where \[b^{(p)}_i \coloneqq g^{(p)}_i - l^{(p)}_i - \sum_{j\in\mathcal{N}_i^+}\hat{p}_{ij} + \sum_{l\in\mathcal{N}_i^-}\breve{p}_{li}.\]
Then, node~$i$ updates $g^{(p)}_i[k]$ and $\lambda_i[k]$ as follows:
\begin{align}\label{pd_sced_un1}
g^{(p)}_i[k+1] &= \left[g^{(p)}_i[k] - s\frac{\partial L^{(i)}[k]}{\partial g^{(p)}_i}\right]_{g^{(p)\min}_i}^{g^{(p)\max}_i},\\
\lambda_i[k+1] &= \lambda_i[k] + s\frac{\partial L^{(i)}[k]}{\partial \lambda_i},
\end{align}
where $L^{(i)}[k]\coloneqq L^{(i)}(x^{(i)}[k])$.
Next, we explain how node~$i$ updates its local flow estimates, $\hat{p}_{ij}[k]$, $(i,j)\in\mathcal{E}_p$, and $\breve{p}_{li}[k]$, $(l,i)\in\mathcal{E}_p$.
%Next, we explain how the estimates of the shared quantities, namely, $\hat{p}_{ij}[k]$ and $\breve{p}_{ij}[k]$, are updated.
Consider $(i,j)\in\mathcal{E}_p$, and note that $\hat{p}_{ij}[k]$ and $\breve{p}_{ij}[k]$ are the estimates of $p_{ij}^*$ maintained by nodes~$i$ and $j$, respectively. To ensure that the estimates~$\hat{p}_{ij}[k]$ and $\breve{p}_{ij}[k]$ converge to the same value, nodes~$i$ and $j$ need to exchange the estimates, compute the average, and use it in their updates as follows:
\begin{align}\label{pd_sced_un2}
%\hat{p}_{ij}[k+1] &= (1-a_{ij}[k])\hat{p}_{ij}[k]+ a_{ij}[k]\breve{p}_{ij}[k] - s\hat{y}^{(p)}_{ij}[k],\\
%\breve{p}_{ij}[k+1] &= (1-a_{ij}[k])\breve{p}_{ij}[k]+ a_{ij}[k]\hat{p}_{ij}[k] - s\breve{y}^{(p)}_{ij}[k],
\hat{p}_{ij}[k+1] &= \Big[(1-a_{ij}[k])\hat{p}_{ij}[k]+ a_{ij}[k]\breve{p}_{ij}[k] \nonumber\\&\quad- s\hat{y}_{ij}[k]\Big]_{p_{ij}^{\min}}^{p_{ij}^{\max}},\\
\breve{p}_{ij}[k+1] &= \Big[(1-a_{ij}[k])\breve{p}_{ij}[k]+ a_{ij}[k]\hat{p}_{ij}[k]\nonumber\\&\quad - s\breve{y}_{ij}[k]\Big]_{p_{ij}^{\min}}^{p_{ij}^{\max}},
\end{align}
where 
\begin{align*}
a_{ij}[k]=\left\{\begin{matrix*}[l]0.5 & \mbox{if }\{i,j\} \in \mathcal{E}_c[k],\\0 & \mbox{otherwise,}\end{matrix*}\right.
\end{align*}
and $\hat{y}_{ij}[k]$ and $\breve{y}_{ij}[k]$ are the estimates of the gradient~$\frac{\partial L}{\partial p_{ij}}$ maintained at nodes~$i$ and $j$, respectively. One way to estimate the gradient can be purely based on the local Lagrangian (local information):
\begin{align}\label{local_grad}
\hat{y}_{ij}[k] &= \frac{\partial L^{(i)}[k]}{\partial \hat{p}_{ij}},\quad
\breve{y}_{ij}[k] = \frac{\partial L^{(j)}[k]}{\partial \breve{p}_{ij}}.
\end{align}
However, leveraging only local information, as in \eqref{local_grad}, results in slow (asymptotic) convergence (to be demonstrated numerically in Section~\ref{subsec:simulations_sced_un}). A better approach is to let each node track the gradient by also using the neighbor's information:
\begin{subequations}\label{grad_track_sced}
\begin{align}
\hat{y}_{ij}[k+1] &= (1-a_{ij}[k]) \hat{y}_{ij}[k] + a_{ij}[k]\breve{y}_{ij}[k]\nonumber\\
&\quad+ 2\left(\frac{\partial L^{(i)}[k+1]}{\partial \hat{p}_{ij}} - \frac{\partial L^{(i)}[k]}{\partial \hat{p}_{ij}}\right), \\
\breve{y}_{ij}[k+1] &= (1-a_{ij}[k]) \breve{y}_{ij}[k] + a_{ij}[k]\hat{y}_{ij}[k]\nonumber\\
&\quad + 2\left(\frac{\partial L^{(j)}[k+1]}{\partial \breve{p}_{ij}} - \frac{\partial L^{(j)}[k]}{\partial \breve{p}_{ij}}\right),
\end{align}
\end{subequations}
with
\begin{align}
\hat{y}_{ij}[0] = 2\frac{\partial L^{(i)}[0]}{\partial \hat{p}_{ij}}, \quad
\breve{y}_{ij}[0] = 2\frac{\partial L^{(j)}[0]}{\partial \breve{p}_{ij}},
\end{align}
where $\hat{y}_{ij}[k]$ and $\breve{y}_{ij}[k]$ are updated so that their average 
%if $\hat{y}_{ij}[k]$ and $\breve{y}_{ij}[k]$ are close to their average, then, 
\begin{align*}
\frac{1}{2}\left(\hat{y}_{ij}[k]+\breve{y}_{ij}[k]\right)&= \frac{\partial L^{(i)}[k]}{\partial \hat{p}_{ij}}+\frac{\partial L^{(j)}[k]}{\partial \breve{p}_{ij}}
\\&=-\lambda_i[k]+\lambda_j[k]
\end{align*}
has exactly the same form as
\begin{align*}
\left.\frac{\partial L}{\partial p_{ij}}\right|_{g^{(p)}[k],p[k],\lambda[k]} &= -\lambda_i[k]+\lambda_j[k],
\end{align*}
used in algorithm~\eqref{centr_pd_sced}. This idea of tracking the gradient, which appeared in \cite{Nedic17} for solving an unconstrained multi-agent optimization problem, allows us to more closely emulate the updates in algorithm~\eqref{centr_pd_sced}, and achieve faster (geometric) convergence rate.

Below, we provide the iterations run by node~$i$:
\begin{subequations}\label{pd_sced_un}
\begin{align}
g^{(p)}_i[k+1] &= \left[g^{(p)}_i[k] - s\frac{\partial L^{(i)}[k]}{\partial g^{(p)}_i}\right]_{g^{(p)\min}_i}^{g^{(p)\max}_i},\label{local_pd_sced_p}\\
\hat{p}_{ij}[k+1] &= \Big[(1-a_{ij}[k])\hat{p}_{ij}[k]+ a_{ij}[k]\breve{p}_{ij}[k] \nonumber\\&\quad- s\hat{y}_{ij}[k]\Big]_{p_{ij}^{\min}}^{p_{ij}^{\max}}, (i,j)\in\mathcal{E}_p,\\
\breve{p}_{li}[k+1] &= \Big[(1-a_{li}[k])\breve{p}_{li}[k]+ a_{li}[k]\hat{p}_{li}[k]\nonumber\\&\quad - s\breve{y}_{li}[k]\Big]_{p_{li}^{\min}}^{p_{li}^{\max}},(l,i)\in\mathcal{E}_p,\\
\hat{y}_{ij}[k+1] &= (1-a_{ij}[k]) \hat{y}_{ij}[k] + a_{ij}[k]\breve{y}_{ij}[k]\nonumber\\
&\quad+ 2\left(\frac{\partial L^{(i)}[k+1]}{\partial \hat{p}_{ij}} - \frac{\partial L^{(i)}[k]}{\partial \hat{p}_{ij}}\right), \\
\breve{y}_{li}[k+1] &= (1-a_{li}[k]) \breve{y}_{li}[k] + a_{li}[k]\hat{y}_{li}[k]\nonumber\\
&\quad + 2\left(\frac{\partial L^{(i)}[k+1]}{\partial \breve{p}_{li}} - \frac{\partial L^{(i)}[k]}{\partial \breve{p}_{li}}\right),\\
\lambda_i[k+1] &= \lambda_i[k] + s\frac{\partial L^{(i)}[k]}{\partial \lambda_i}.\label{local_pd_sced_lmbd}
\end{align}
\end{subequations}
%In the following, we expand on the ideas of Section~\ref{} use exactly the same ideas to design a distributed algorithm for solving rSOCP.

\subsection{Feedback Interconnection Representation of the Distributed Primal-Dual Algorithm}\label{subsec:feedback_sced}
In the following, we represent~\eqref{pd_sced_un} as a feedback interconnection of a nominal system, denoted by $\mathcal{H}_1$, and a disturbance system, denoted by $\mathcal{H}_2$, which allows us to utilize the small-gain theorem for convergence analysis purposes. To this end, let
\begin{align*}
g^{(p)}[k]&\coloneqq[g^{(p)}_1[k],\dots,g^{(p)}_n[k]]^\top, \lambda[k]\coloneqq[\lambda_1[k],\dots,\lambda_n[k]]^\top,\\
\hat{p}[k]&\coloneqq [\{\hat{p}_{ij}[k]\}_{(i,j)\in\mathcal{E}_p}], \breve{p}[k]\coloneqq [\{\breve{p}_{ij}[k]\}_{(i,j)\in\mathcal{E}_p}],\\
\hat{y}[k]&\coloneqq [\{\hat{y}_{ij}[k]\}_{(i,j)\in\mathcal{E}_p}], \breve{y}[k]\coloneqq [\{\breve{y}_{ij}[k]\}_{(i,j)\in\mathcal{E}_p}],\\
\overline{p}[k]&\coloneqq \frac{1}{2}(\hat{p}[k]+\breve{p}[k]), \overline{y}[k] \coloneqq \frac{1}{2}(\hat{y}[k]+\breve{y}[k]). 
\end{align*}
By using \eqref{pd_sced_un}, we write the iterations for $g^{(p)}[k]$, $\overline{p}[k]$, and $\lambda[k]$, which constitute the nominal system, $\mathcal{H}_1$, given by:
\begin{subequations}\label{H1_sced_un}
\begin{align}
{\mathcal{H}_1:}\mbox{ }g^{(p)}[k+1] &= \left[g^{(p)}[k] - s\frac{\partial \overline{L}[k]}{\partial g^{(p)}}+e_g[k]\right]_{g^{(p)\min}}^{g^{(p)\max}},\\
\overline{p}[k+1] &= \frac{1}{2}\left[\overline{p}[k] - s\frac{\partial \overline{L}[k]}{\partial p}+e_p[k]\right]_{p^{\min}}^{p^{\max}}\nonumber\\
&\quad+\frac{1}{2}\left[\overline{p}[k] - s\frac{\partial \overline{L}[k]}{\partial p}-e_p[k]\right]_{p^{\min}}^{p^{\max}},\label{H1_sced_un_2}\\
\lambda[k+1] &= \lambda[k] + s\frac{\partial \overline{L}[k]}{\partial \lambda} + e_{\lambda}[k],
\end{align}
\end{subequations}
where
\begin{align*}
\overline{L}[k] &\coloneqq L(g^{(p)}[k],\overline{p}[k],\lambda[k]),\\
e_g[k] &\coloneqq -s\rho (M\overline{p}[k]-M_0\hat{p}[k]-N_0\breve{p}[k]),\\
e_p[k] &\coloneqq \hat{p}[k]-\overline{p}[k] + s\overline{y}[k]-s\hat{y}[k],\\
e_{\lambda}[k] &\coloneqq s(M\overline{p}[k]-M_0\hat{p}[k]-N_0\breve{p}[k]), 
\end{align*}
to obtain~\eqref{H1_sced_un_2}, we used the fact that
\[\frac{\partial \overline{L}[k]}{\partial p} = \overline{y}[k].\]
We note that $e[k]\coloneqq [e_g[k]^\top,e_p[k]^\top,e_{\lambda}[k]^\top]^\top$ results from $(\hat{p}[k],\breve{p}[k])$ and $(\hat{y}[k],\breve{y}[k])$ deviating from their respective average, $\overline{p}[k]$ and $\overline{y}[k]$; without $e[k]$, the nominal system~$\mathcal{H}_1$ has exactly the same form as~\eqref{centr_pd_sced}.
Now, we define the disturbance system, $\mathcal{H}_2$, as follows: 
\begin{subequations}\label{H2_sced_un}
\begin{align}
{\mathcal{H}_2:}\mbox{ }\hat{p}_{ij}[k+1] &= \Big[(1-a_{ij}[k])\hat{p}_{ij}[k]+ a_{ij}[k]\breve{p}_{ij}[k] \nonumber\\&\quad- s\hat{y}_{ij}[k]\Big]_{p_{ij}^{\min}}^{p_{ij}^{\max}},(i,j)\in\mathcal{E}_p,\\
\breve{p}_{ij}[k+1] &= \Big[(1-a_{ij}[k])\breve{p}_{ij}[k]+ a_{ij}[k]\hat{p}_{ij}[k]\nonumber\\&\quad - s\breve{y}_{ij}[k]\Big]_{p_{ij}^{\min}}^{p_{ij}^{\max}},\\
\hat{y}_{ij}[k+1] &= (1-a_{ij}[k]) \hat{y}_{ij}[k] + a_{ij}[k]\breve{y}_{ij}[k]\nonumber\\
&\quad+ 2\left(\frac{\partial L^{(i)}[k+1]}{\partial \hat{p}_{ij}} - \frac{\partial L^{(i)}[k]}{\partial \hat{p}_{ij}}\right), \\
\breve{y}_{ij}[k+1] &= (1-a_{ij}[k]) \breve{y}_{ij}[k] + a_{ij}[k]\hat{y}_{ij}[k]\nonumber\\
&\quad + 2\left(\frac{\partial L^{(j)}[k+1]}{\partial \breve{p}_{ij}} - \frac{\partial L^{(j)}[k]}{\partial \breve{p}_{ij}}\right),\\
e_g[k] &= -s\rho (M\overline{p}[k]-M_0\hat{p}[k]-N_0\breve{p}[k]),\\
e_p[k] &= \hat{p}[k]-\overline{p}[k] + s\overline{y}[k]-s\hat{y}[k],\\
e_{\lambda}[k] &= s(M\overline{p}[k]-M_0\hat{p}[k]-N_0\breve{p}[k]).
\end{align}
\end{subequations}
Then, as illustrated in Fig.~\ref{fig:feedback_sced}, algorithm~\eqref{pd_sced_un} can be viewed as a feedback interconnection of $\mathcal{H}_1$ and $\mathcal{H}_2$.
%where $(g^{(p)*},p^*,\lambda^*)$ is the equilibrium of~\eqref{H1_sced_un} when $e[k]\equiv 0$, for all $k\geq0$.
%This perspective allows us to carry out the convergence analysis by using the small-gain theorem study the convergence  
\begin{figure}
    \centering 
	\includegraphics[trim=0cm 0cm 0cm 0cm, clip=true, scale=0.6]{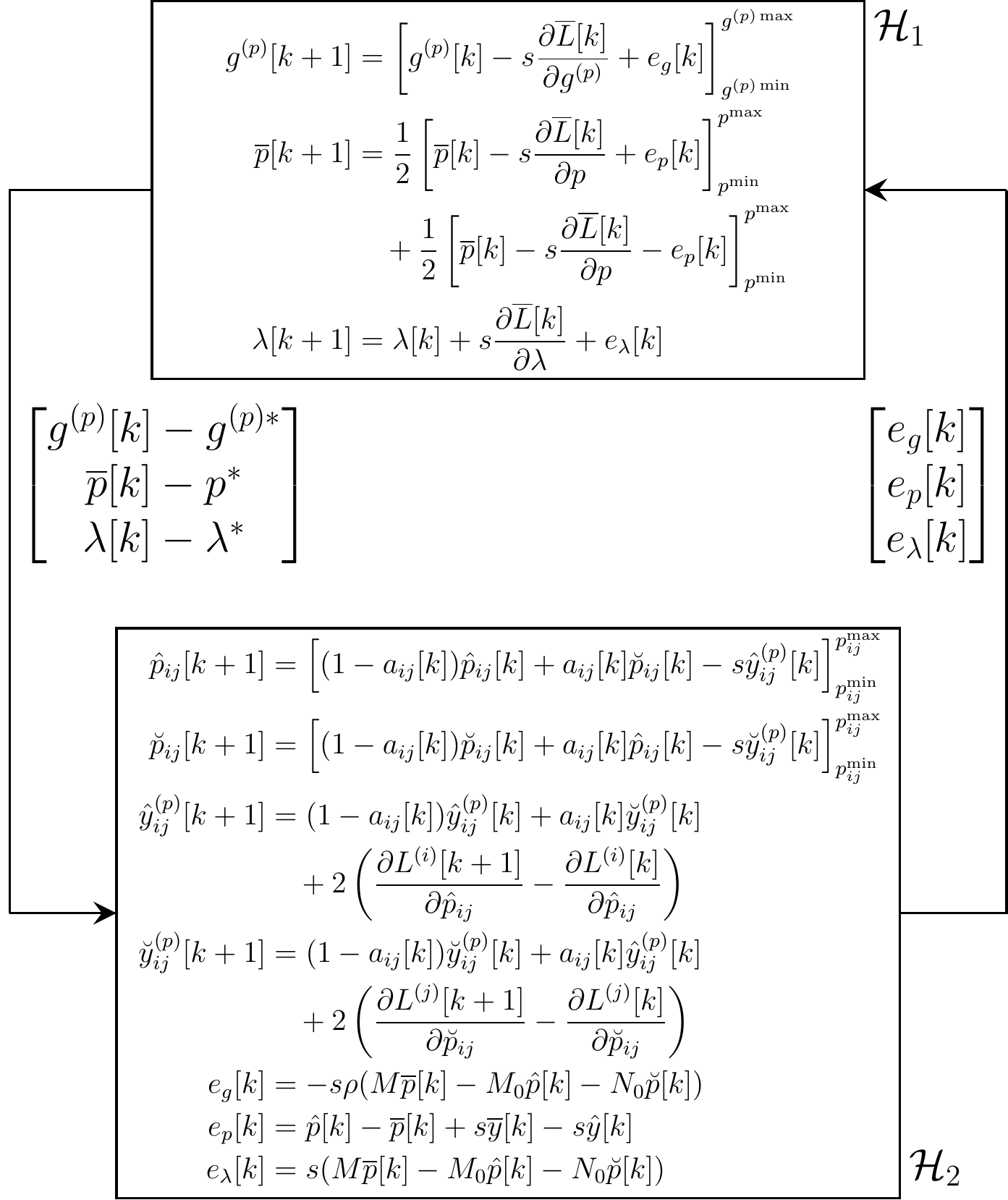} 
    \caption{Algorithm~\eqref{pd_sced_un} as a feedback system.}
    \vspace{-5pt} 
    \label{fig:feedback_sced}
\end{figure}
In the following, we find the relationship between the loop gain of the feedback system and the stepsize~$s$ establishing that the loop gain can be reduced by decreasing $s$. This enables us to apply the small-gain theorem and show that the feedback loop does not amplify the energy of the convergence error, but, on the contrary, the error eventually decays to zero, if the loop gain is sufficiently small.

\subsection{Convergence Analysis}\label{subsec:conv_analysis_sced_un}
%We quantify the effect of the feedback system on the convergence error in terms of the stepsize
We first establish that the systems~$\mathcal{H}_1$ and $\mathcal{H}_2$ are finite-gain stable:
\begin{itemize}
\item[{\textbf{R1.}}] $\|z\|_2^{a,K}\leq \alpha_1\|e\|_2^{a,K}+\beta_1$ for some positive $\alpha_1$ and $\beta_1$,
\item[{\textbf{R2.}}] $\|e\|_2^{a,K}\leq s\alpha_2\|z\|_2^{a,K}+\beta_2$ for some positive $\alpha_2$ and $\beta_2$,
\end{itemize}
for some $a \in (0,1)$ and sufficiently small~$s>0$,
where
\begin{align}
z[k] \coloneqq \begin{bmatrix}g^{(p)}[k]-g^{(p)*}\\\overline{p}[k]-p^*\\ \lambda[k]-\lambda^*\\ \end{bmatrix}
\end{align}
denotes the convergence error.
Then, from {\textbf{R1}} and {\textbf{R2}}, we determine that the loop gain is $s\alpha_1\alpha_2$. Noticing that the gain~$s\alpha_1\alpha_2$ becomes strictly smaller than $1$ for sufficiently small~$s$, we apply the small-gain theorem and show that $\|z\|_2^{a,K}$ is bounded for all $K>0$ resulting in the convergence of $z[k]$ to zero at a geometric rate~$\mathcal{O}(a^k)$. 
In the next results, we establish that the relations {\textbf{R1}} and {\textbf{R2}} hold and that algorithm~\eqref{pd_sced_un} converges geometrically fast.
[We omit the proofs, since they are analogous to those of similar results established in Section~\ref{sec:opf} dealing with rSOCP.]
\begin{proposition}\label{prop:H1_sced_un}
Let Assumption~\ref{objective_assumption} hold.
Then, under~\eqref{H1_sced_un}, we have that
\begin{align}
\textnormal{\textbf{R1. }} \|z\|_2^{a,K}\leq \alpha_1\|e\|_2^{a,K}+\beta_1,
\label{z_to_e_sced_un}
\end{align}
for some positive~$\alpha_1$ and $\beta_1$, $a \in (0,1)$, and sufficiently small~$s>0$.
%If $e[k]\equiv 0$, then,~\eqref{H1_dd} converges to $(p^*,\lambda^*)$.
\end{proposition}
\begin{proposition}\label{prop:H2_sced_un}
Let Assumptions~\ref{objective_assumption} and \ref{assume_comm_model_un} hold.
Then, under~\eqref{H2_sced_un}, we have that
\begin{align}\label{e_to_z_sced_un}
\textnormal{\textbf{R2. }}\|e\|_2^{a,K}\leq s\alpha_2\|z\|_2^{a,K}+\beta_2,
\end{align}
for some positive~$\alpha_2$ and $\beta_2$, $a \in (0,1)$, and sufficiently small~$s>0$.
\end{proposition}
\begin{proposition}\label{prop:small_gain_sced_un}
Let Assumptions~\ref{objective_assumption} and \ref{assume_comm_model_un} hold.
Then, under algorithm~\eqref{pd_sced_un}, we have that
\begin{align}
\|z\|_2^{a,K}\leq \beta,
\label{z_relation_sced_un}
\end{align}
for some $a \in (0,1)$, $\beta>0$, and sufficiently small~$s>0$. In particular, $x^{(i)}[k]$ converges to $x^{(i)*}$, $\forall i$, at a geometric rate~$\mathcal{O}(a^k)$.
\end{proposition}
Finally, we show that $(g^{(p)*},p^*)$ is the solution of SCED. 
\begin{lemma}\label{lem:sced_sol}
Consider $(g^{(p)*},p^*,\lambda^*)$, namely, the equilibrium of \eqref{centr_pd_sced}.
Then, $(g^{(p)*},p^*)$ is the solution of SCED.
\end{lemma}
\begin{proof}
At the equilibrium, we have that
\begin{align*}
g^{(p)*} &= \left[g^{(p)*} - s\nabla f(g^{(p)*}) - s\lambda^* \right]_{g^{(p)\min}}^{g^{(p)\max}},\\
p^* &= \left[p^* + sM^\top\lambda^*\right]_{p^{\min}}^{p^{\max}},\\
\lambda^* &= \lambda^* + s(g^{(p)*}- l^{(p)}-Mp^*).
\end{align*}
Then, the following relations hold:
\begin{subequations}\label{KKT}
\begin{align}
0 &= \nabla f(g^{(p)*})+\lambda^* + \mu^* - \nu^*,\\
0 &= M^\top\lambda^*-\alpha^*+\beta^*,\\
0 &= g^{(p)*}- l^{(p)}-Mp^*,\\
0 &=\mu_i^*(g^{(p)*}_i-g^{(p)\max}_i),\\
0 &=\nu_i^*(g^{(p)\min}_i-g^{(p)*}_i),i=1,\dots,n,\\
0 &=\alpha_{ij}^*(p_{ij}^*-p^{\max}_{ij}),\\
0 &=\beta_{ij}^*(p^{\min}_{ij}-p_{ij}^*), (i,j)\in\mathcal{E}_p,
\end{align}
\end{subequations}
for some non-negative $\mu_i^*$, $\nu_i^*$, $i=1,\dots,n$, $\alpha_{ij}^*$, and $\beta_{ij}^*$, $(i,j)\in\mathcal{E}_p$, where $\mu^*=[\mu_1^*,\dots,\mu_n^*]^\top$, $\nu^*=[\nu_1^*,\dots,\nu_n^*]^\top$, $\alpha^* = [\{\alpha_{ij}^*\}_{(i,j)\in\mathcal{E}_p}]$, and $\beta^* = [\{\beta_{ij}^*\}_{(i,j)\in\mathcal{E}_p}]$. Noticing that \eqref{KKT} represents the Karush-Kuhn-Tucker (KKT) conditions for SCED, it follows from \cite[Proposition~3.3.1]{NonlinearProgramming} that $(g^{(p)*},p^*)$ is the solution of SCED.
\end{proof}

\subsection{Numerical Simulations}
\label{subsec:simulations_sced_un}
Next, we present numerical results to illustrate the performance of the distributed primal-dual algorithm~\eqref{pd_sced_un} over undirected graph $\mathcal{G}^{(c)}[k]$ using the IEEE~$69$--bus radial test system \cite{Matpower}. 

%using the IEEE~$39$--bus test system \cite{Matpower}. 
%\begin{figure}
%    \centering\vspace{3pt} 
%	\includegraphics[trim=3.3cm 0cm 0cm 0.1cm, clip=true, scale=0.28]{Figures/IEEE39.png} 
%	\vspace{-7pt} 
%    \caption{IEEE~$39$-bus test system.}
%    \vspace{-5pt} 
%    \label{fig:39-bus}
%\end{figure}
A subset of buses are designated to have a DER. For a DER at bus~$i$, we choose $f_i(p_i) = a_ip_i^2$, where $a_i>0$ is randomly selected. It is assumed that each communication link becomes inactive with probability~$0.4$. The algorithm uses a constant stepsize~$s = 2\times10^{-2}$ and $\rho=2$. For initialization, we use $g^{(p)}[0] = l_p[0]$, $\hat{p}_{ij}[0]=0$, $\breve{p}_{ij}[0]=0$, $(i,j)\in\mathcal{E}_p$, and 
\[\hat{y}_{ij}[0] = 2\frac{\partial L^{(i)}[0]}{\partial \hat{p}_{ij}}, \quad
\breve{y}_{ij}[0] = 2\frac{\partial L^{(j)}[0]}{\partial \breve{p}_{ij}}.\]

In the distributed implementation, communicating data takes much longer than one iteration executed by a computing device. Rather than the total number of iterations, the number of communication attempts can serve as a more appropriate performance metric to evaluate the practical usefulness of the algorithm.
We believe that it is reasonable to assume that a computing device is able to perform a number of iterations (less than 100) between consecutive communication attempts. Let $m$ denote the number of iterations between consecutive communication attempts. In the numerical example, we used different values of $m$. We note that making $m$ large or even finding a minimum of the local Lagrangians, $L^{(i)}(x^{(i)})$, $i\in\mathcal{V}_p$, does not necessarily make the performance better. On the contrary, keeping $m$ relatively small ($m<20$) often achieves a much better performance.

In Fig.~\ref{fig:num_results_sced}, we compare the performance of algorithm~\eqref{pd_sced_un}, for convenience referred to as~$\mathbf{A}_1$, 
where we recall that
\begin{subequations}\label{grad_track_sced_rpt}
\begin{align}
\hat{y}_{ij}[k+1] &= (1-a_{ij}[k]) \hat{y}_{ij}[k] + a_{ij}[k]\breve{y}_{ij}[k]\nonumber\\
&\quad+ 2\left(\frac{\partial L^{(i)}[k+1]}{\partial \hat{p}_{ij}} - \frac{\partial L^{(i)}[k]}{\partial \hat{p}_{ij}}\right), \\
\breve{y}_{ij}[k+1] &= (1-a_{ij}[k]) \breve{y}_{ij}[k] + a_{ij}[k]\hat{y}_{ij}[k]\nonumber\\
&\quad + 2\left(\frac{\partial L^{(j)}[k+1]}{\partial \breve{p}_{ij}} - \frac{\partial L^{(j)}[k]}{\partial \breve{p}_{ij}}\right),
\end{align}
\end{subequations}
against that of algorithm~\eqref{pd_sced_un1}--\eqref{local_grad}, referred to as~$\mathbf{A}_2$, where 
\begin{align}\label{local_grad_rpt}
\hat{y}_{ij}[k] &= \frac{\partial L^{(i)}[k]}{\partial \hat{p}_{ij}},\quad
\breve{y}_{ij}[k] = \frac{\partial L^{(j)}[k]}{\partial \breve{p}_{ij}}.
\end{align}
We demonstrate that updating $\hat{y}_{ij}[k]$ and $\breve{y}_{ij}[k]$ using \eqref{grad_track_sced_rpt} results in a much better performance than using \eqref{local_grad_rpt}. 
Figure~\ref{fig:num_results_sced} shows the trajectory of the relative cost error of the obtained solutions, namely, \[\frac{|f(g^{(p)}[k])-f(g^{(p)*})|}{f(g^{(p)*})},\] and the evolution of the largest constraint violation. In both algorithms, each node runs $m=5$ iterations between consecutive communication attempts. The results in Fig.~\ref{fig:num_results_sced} demonstrate that $\mathbf{A}_1$ has geometric convergence rate, and converges significantly faster than $\mathbf{A}_2$. We note that $\mathbf{A}_2$ has asymptotic convergence rate and requires stepsize $s$ to go to zero asymptotically. In the simulations, $\mathbf{A}_2$ uses $s[k] = a/(k+b)$, with $a = 1$, and $b=10$.
\begin{figure}
    \begin{subfigure}[t]{0.5\textwidth}
        \centering
        \includegraphics[trim=0cm 0cm 0cm 0cm, clip=true, scale=0.45]{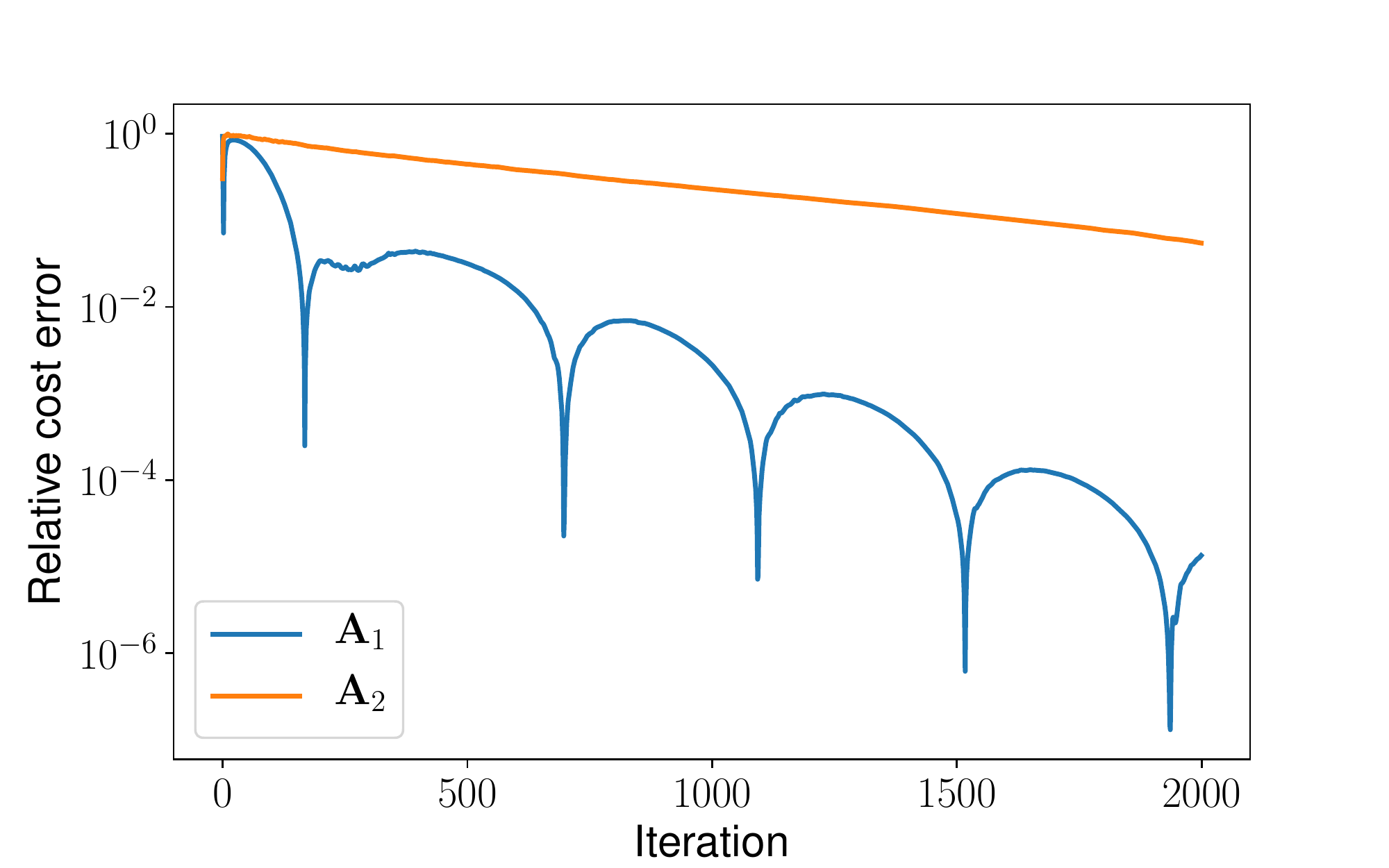}
%        \vspace{-15pt}
        \caption{Trajectories of the relative cost error for the estimated solutions}
        \label{fig:cost-sced}
    \end{subfigure}\vspace{10pt}\\
    \begin{subfigure}[t]{0.5\textwidth}
        \centering
        \includegraphics[trim=0cm 0cm 0cm 0cm, clip=true, scale=0.45]{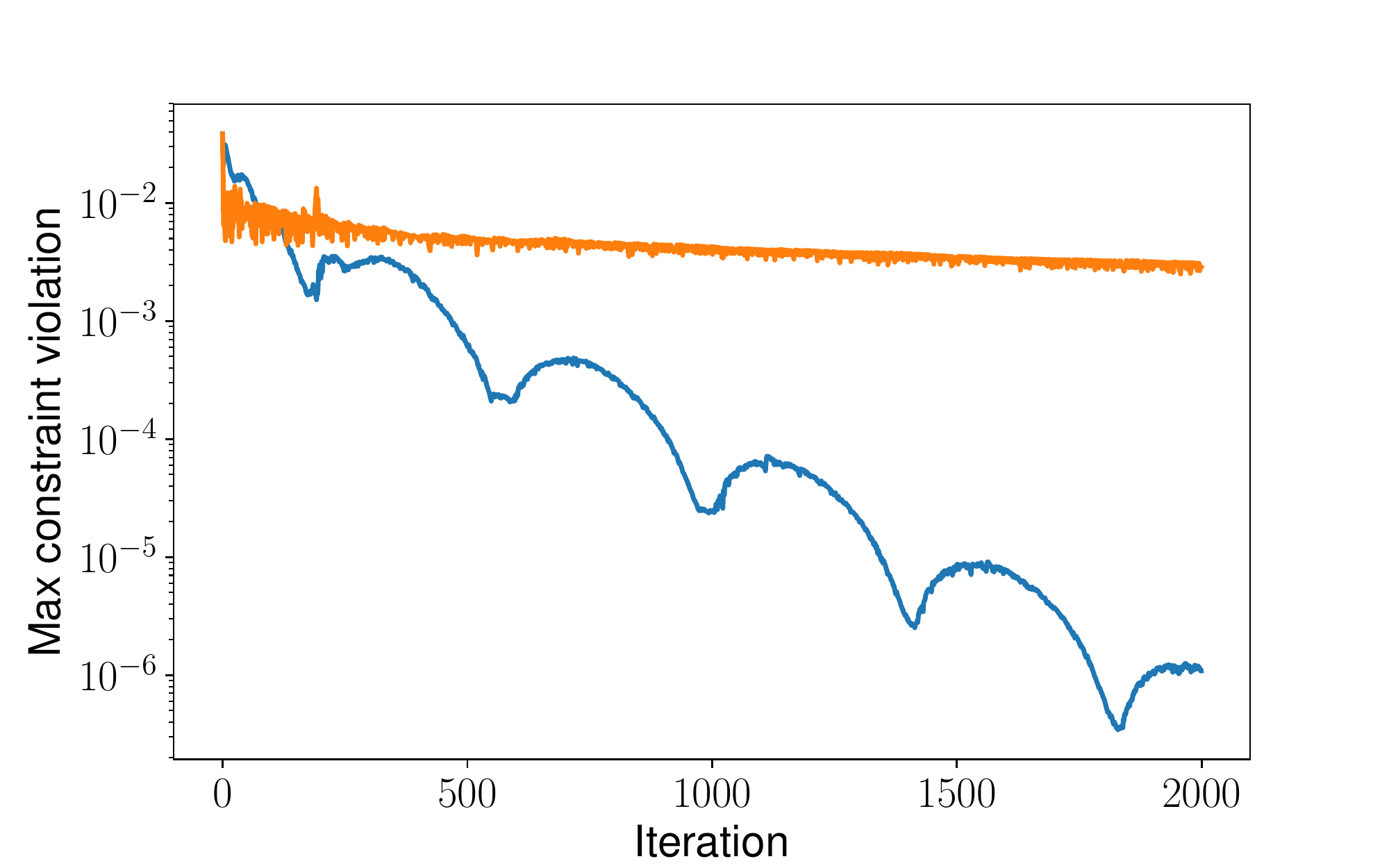}
%        \vspace{-15pt}
        \caption{Evolution of the largest constraint violation}
        \label{fig:constr-sced}
    \end{subfigure}
    \vspace{10pt}
    \caption{A comparison of the performances of $\mathbf{A}_1$ and $\mathbf{A}_2$.}
    \label{fig:num_results_sced}
\end{figure}

\subsection{Time-Varying Directed Communication Graphs}\label{subsec:sced_dd}
When $\mathcal{G}^{(c)}[k]$ is directed, the iterations in \eqref{pd_sced_un} fail to converge. In the following, we provide the robust extension of \eqref{pd_sced_un} that solves SCED over time-varying directed graphs.

In the robust extension, we let node~$i$ perform the same updates for $g^{(p)}_i$ and $\lambda_i$ as the ones in \eqref{pd_sced_un}, namely,
\begin{subequations}\label{pd_sced_dd0}
\begin{align}
g^{(p)}_i[k+1] &= \left[g^{(p)}_i[k] - s\frac{\partial L^{(i)}[k]}{\partial g^{(p)}_i}\right]_{g^{(p)\min}_i}^{g^{(p)\max}_i},\\
\lambda_i[k+1] &= \lambda_i[k] + s\frac{\partial L^{(i)}[k]}{\partial \lambda_i}.
\end{align}
\end{subequations}
The only difference between \eqref{pd_sced_un} and its robust extension is how the averaging step is performed in the updates of the flow and gradient estimates. The key idea behind the averaging step in the robust extension is to let the neighboring nodes perform averaging exactly once over possibly longer time periods. In other words, we ensure that, for any given $(i,j)\in\mathcal{E}_p$, there exists a sequence of time instants~$\{T_k\}_{k=0}^\infty$ such that nodes~$i$ and $j$ perform averaging (using the values they received from each other) exactly once during each time interval~$(T_k,T_{k+1})$, $k=0,1,\dots$.
One of the simplest ways to implement such strategy over directed $\mathcal{G}^{(c)}[k]$ is to let nodes~$i$ and $j$ perform averaging in an alternating fashion. In other words, once node~$i$ performs averaging at time instant~$t_k$ using the value received from node~$j$, it waits for node~$j$ to perform averaging (before node~$i$ performs averaging again). While in this waiting mode, node~$i$ still performs a local update and, possibly, averaging with other neighbors. To implement this strategy over directed $\mathcal{G}^{(c)}[k]$, nodes~$i$ and $j$ need to maintain and communicate certain acknowledgement flags. By sending a flag, a node intends to let its neighbor know whether or not it has performed averaging; then, based on this information, the neighbor decides whether or not it should perform averaging. Below, we provide more details of this approach.

Suppose node~$i$ receives $\breve{p}_{ij}[k]$ and the flag from node~$j$ at time instant~$t_k$. If the status of the received flag is different from the previously received one, then, node~$i$ flips the status of its own flag, stores $\hat{p}_{ij}[k]$ and $\breve{p}_{ij}[k]$, and performs averaging as follows:
\begin{align}
\hat{p}_{ij}[k+1] &= \Big[(1-a_{ij}[k])\hat{p}_{ij}[k]+ a_{ij}[k]\breve{p}_{ij}[k]\nonumber\\
&\quad  - s\hat{y}_{ij}[k]\Big]_{p_{ij}^{\min}}^{p_{ij}^{\max}},
\end{align} 
Over the next iterations, it keeps sending $\hat{p}_{ij}[k]$, $\breve{p}_{ij}[k]$ and its flag to node~$j$, until node~$i$ receives a different flag from node~$j$. Meanwhile, if node~$j$ receives a different flag (from the previously received one) from node~$i$ (which means that node~$i$ has performed averaging) at some time instant~$t_\tau> t_k$, node~$j$ flips the status of its own flag and performs averaging, but slightly differently, as follows:
\begin{align}
%\hat{p}_{ij}[k+1] &= (1-a_{ij}[k])\hat{p}_{ij}[k]+ a_{ij}[k]\breve{p}_{ij}[k] - s\hat{y}^{(p)}_{ij}[k],\\
\breve{p}_{ij}[\tau+1] &= \Big[(1-a_{ij}[k])\breve{p}_{ij}[k]+ a_{ij}[k]\hat{p}_{ij}[k]\nonumber\\&\quad + (\breve{p}_{ij}[\tau] - \breve{p}_{ij}[k]) - s\breve{y}_{ij}[\tau\Big]_{p_{ij}^{\min}}^{p_{ij}^{\max}}.
\label{j_ave}
\end{align} 
Note that, in the averaging step, node~$j$ uses the same values that node~$i$ used at time instant~$t_k$. In \eqref{j_ave}, we also have $(\breve{p}_{ij}[\tau] - \breve{p}_{ij}[k])$, which is the sum of all gradient terms, $s\breve{y}_{ij}[t]$, $t=k,k+1,\dots,\tau$, that have been accumulated since time instant~$t_k$. In this way, we are able to mimic the corresponding iteration in algorithm \eqref{pd_sced_un}. Also, note that if nodes~$i$ and $j$ happen to perform averaging within the same time period~$(t_k,t_{k+1})$, then, $\tau = k$, and \eqref{j_ave} exactly mimics the corresponding iteration in algorithm \eqref{pd_sced_un}. This scheme, where nodes perform averaging in alternating fashion, is referred to as the \textit{alternating averaging protocol} (see, e.g., \cite{ZhDo18,ZhDo19}). We now formally define the updates of the flow and gradient estimates as follows.
For each $(i,j)\in\mathcal{E}_p$, node~$i$ runs the following iterations: 
\begin{subequations}\label{pd_sced_dd1}
\begin{align}
\hat{p}_{ij}[k+1] &= \Big[(1-a_{ij}[k])\hat{p}_{ij}[k]+ a_{ij}[k]\breve{p}_{ij}[k]\nonumber\\
&\quad  - s\hat{y}_{ij}[k]\Big]_{p_{ij}^{\min}}^{p_{ij}^{\max}},\\ 
\hat{y}_{ij}[k+1] &= (1-a_{ij}[k]) \hat{y}_{ij}[k] + a_{ij}[k]\breve{y}_{ij}[k]\nonumber\\
&\quad + 2\left(\frac{\partial L^{(i)}[k+1]}{\partial \hat{w}_{ij}} - \frac{\partial L^{(i)}[k]}{\partial \hat{w}_{ij}}\right), 
\end{align}
\end{subequations}
while node~$j$ executes the iterations given below:
\begin{subequations}\label{pd_sced_dd2}
\begin{align}
\breve{p}_{ij}[k+1] &= \Big[(1-a_{ji}[k])\breve{r}_{ij}[k]+ a_{ji}[k]\hat{r}_{ij}[k]\nonumber\\
&\quad+\breve{p}_{ij}[k] - \breve{r}_{ij}[k]- s\breve{y}_{ij}[k]\Big]_{p_{ij}^{\min}}^{p_{ij}^{\max}},\\
\breve{y}_{ij}[k+1] &= (1-a_{ij}[k]) \breve{\rho}_{ij}[k] + a_{ij}[k]\hat{\rho}_{ij}[k]\nonumber\\&\quad+\breve{y}_{ij}[k] - \breve{\rho}_{ij}[k]\nonumber\\
&\quad + 2\left(\frac{\partial L^{(i)}[k+1]}{\partial \breve{w}_{ij}} - \frac{\partial L^{(i)}[k]}{\partial \breve{w}_{ij}}\right),
\end{align}
\end{subequations}
where $a_{ij}[k]$, $a_{ji}[k]$, $\hat{r}_{ij}[k]$, $\breve{r}_{ij}[k]$, $\hat{\rho}_{ij}[k]$, and $\breve{\rho}_{ij}[k]$ are updated using the \textit{alternating averaging protocol} (see, e.g., \cite{ZhDo18,ZhDo19}):
\begin{subequations}\label{alt-ave}
\begin{align*}
\mbox{node~$i$:}\\
\breve{\phi}_{ij}[k] &= \left\{\begin{array}{l l}\phi_{ji}[k] &\mbox{if } (j,i)\in\vec{\mathcal{E}}_c[k],\\ \breve{\phi}_{ij}[k-1] & \mbox{otherwise,}\end{array}\right.\\
\phi_{ij}[k] &= \left\{\begin{array}{l l}\neg\phi_{ij}[k-1] &\mbox{if } \breve{\phi}_{ij}[k]\neq \breve{\phi}_{ij}[k-1],\\ \phi_{ij}[k-1] & \mbox{otherwise,}\end{array}\right.\\
a_{ij}[k] &= \left\{\begin{array}{l l}0.5 &\mbox{if } (j,i)\in\vec{\mathcal{E}}_c[k], \breve{\phi}_{ij}[k]\neq \breve{\phi}_{ij}[k-1],\\ 0 & \mbox{otherwise,}\end{array}\right.\\
\mbox{node~$j$:}\\
\hat{\phi}_{ij}[k] &= \left\{\begin{array}{l l}\phi_{ij}[k]&\mbox{if } (i,j)\in\vec{\mathcal{E}}_c[k],\\ \hat{\phi}_{ij}[k-1] & \mbox{otherwise,}\end{array}\right.\\
\phi_{ji}[k] &= \left\{\begin{array}{l l}\neg\phi_{ji}[k-1] &\mbox{if } \hat{\phi}_{ij}[k]\neq \hat{\phi}_{ij}[k-1],\\ \phi_{ji}[k-1] & \mbox{otherwise,}\end{array}\right.\\
a_{ji}[k] &= \left\{\begin{array}{l l}0.5 &\mbox{if } (i,j)\in\vec{\mathcal{E}}_c[k], \hat{\phi}_{ij}[k]\neq \hat{\phi}_{ij}[k-1],\\ 0 & \mbox{otherwise,}\end{array}\right.\\
\hat{r}_{ij}[k] &= \left\{\begin{array}{l l}\hat{p}_{ij}[t_k] &\mbox{if } (i,j)\in\vec{\mathcal{E}}_c[k], \hat{\phi}_{ij}[k]\neq \hat{\phi}_{ij}[k-1],\\ 0 & \mbox{otherwise,}\end{array}\right.\\
\breve{r}_{ij}[k] &= \left\{\begin{array}{l l}\breve{p}_{ij}[t_k] &\mbox{if } (i,j)\in\vec{\mathcal{E}}_c[k], \breve{\phi}_{ij}[k]\neq \breve{\phi}_{ij}[k-1],\\ 0 & \mbox{otherwise,}\end{array}\right.\\
\hat{\rho}_{ij}[k] &= \left\{\begin{array}{l l}\hat{y}_{ij}[t_k] &\mbox{if } (i,j)\in\vec{\mathcal{E}}_c[k], \hat{\phi}_{ij}[k]\neq \hat{\phi}_{ij}[k-1],\\ 0 & \mbox{otherwise,}\end{array}\right.\\
\breve{\rho}_{ij}[k] &= \left\{\begin{array}{l l}\breve{y}_{li}[t_k] &\mbox{if } (i,j)\in\vec{\mathcal{E}}_c[k], \breve{\phi}_{ij}[k]\neq \breve{\phi}_{ij}[k-1],\\ 0 & \mbox{otherwise,}\end{array}\right.
\tag{\ref{alt-ave}}
\end{align*}
\end{subequations}
where $\neg$ denotes the logical negation, i.e., $\neg \xi = 1$ if $\xi=0$, and $\neg \xi = 0$, otherwise; $t_k\leq k$ denotes the latest time, when node~$i$ performed averaging. In protocol \eqref{alt-ave}, $\phi_{ij}$ and $\phi_{ji}$ are the acknowledgement flags maintained by nodes~$i$ and $j$, respectively. Nodes~$i$ and $j$ store the received statuses of the flags~$\phi_{ji}$ and $\phi_{ij}$ in $\breve{\phi}_{ij}$ and $\hat{\phi}_{ij}$, respectively. 

Initially, $\phi_{ji}[0] = 1$, $\phi_{ij}[0] = 0$, $\breve{\phi}_{ij}[0] = 0$, and $\hat{\phi}_{ij}[0] = 0$. The reason for setting the node~$j$'s flag, $\phi_{ji}[0]$, to $1$ is to initiate the protocol execution. [If both flags, $\phi_{ij}$ and $\phi_{ji}$, are set to zero, the protocol will never execute.] 
Below, we state the convergence result for the robust primal-dual algorithm \eqref{pd_sced_dd0}, \eqref{pd_sced_dd1}, \eqref{pd_sced_dd2}, and \eqref{alt-ave}, omitting the proof since it is analogous to that of a similar result established in Section~\ref{sec:opf} dealing with rSOCP.
\begin{proposition}\label{prop:small_gain_sced_dd}
Let Assumptions~\ref{objective_assumption} and \ref{assume_comm_model_dd} hold.
Then, under algorithm~\eqref{pd_sced_dd0}, \eqref{pd_sced_dd1}, \eqref{pd_sced_dd2}, and \eqref{alt-ave}, we have that
\begin{align}
\|z\|_2^{a,K}\leq \beta,
\label{z_relation_sced_dd}
\end{align}
for some $a \in (0,1)$, $\beta>0$, and sufficiently small~$s>0$. In particular, $x^{(i)}[k]$ converges to $x^{(i)*}$, $\forall i$, at a geometric rate~$\mathcal{O}(a^k)$.
\end{proposition}

\section{Distributed OPF Over Time-Varying Communication Graphs}\label{sec:opf}
In this section, we expand on the ideas of Section~\ref{sec:sced} and present the distributed primal-dual algorithms for solving rSOCP over time-varying communication graphs. We first consider the case of undirected communication graphs. Then, we tackle the general case of directed communication graphs.
\subsection{Time-Varying Undirected Communication Graphs}
Let $x \coloneqq [g^{(p)},g^{(q)},v]^\top$, and $\omega \coloneqq [p,q,\varepsilon,\ell]^\top$, and let $\gamma \coloneqq [\lambda,\mu,\nu,\eta]^\top$ and $\tau$ denote the dual variables associated with the DistFlow model constraints \eqref{DistFlow1}--\eqref{DistFlow2}, \eqref{DistFlow5}, and \eqref{DistFlow6} in rSOCP. Let $L(x,\gamma,\tau)$ denote the augmented Lagrangian for rSOCP given by
\begin{align*}
L(x,\omega,\gamma,\tau)  &= f(g^{(p)}) + \lambda^\top b^{(p)} + \mu^\top b^{(q)} + \nu^\top b^{(v)} \\
&\quad+ \eta^\top(\varepsilon - M^\top v)+\rho \|\ell\|_2^2+ \rho\|v-1\|_2^2\\
&\quad+\tau^\top(p\circ p+q\circ q - M_0^\top v\circ \ell)\\
&\quad + \rho_1 \|b^{(p)}\|_2^2 + \rho_2 \|b^{(q)}\|_2^2 + \rho_3 \|b^{(v)}\|_2^2,
\end{align*}
where $b^{(v)}$, $b^{(p)}$ and $b^{(q)}$ are defined as follows:
\begin{align*}
b^{(v)} &\coloneqq \varepsilon - 2Rp-2Xq+(R^2+X^2)\ell,\\
b^{(p)} &\coloneqq g^{(p)} - l^{(p)} - Mp + N_0R\ell,\\
b^{(q)} &\coloneqq g^{(q)} - l^{(q)} - Mq + N_0X\ell.
\end{align*}
The regularization terms~$\rho_1 \|b^{(p)}\|_2^2$, $\rho_2 \|b^{(q)}\|_2^2$, and $\rho_3 \|b^{(v)}\|_2^2$ penalize the violation of the constraints and allow us to significantly improve the convergence speed.

Our starting point to solve rSOCP is the following primal-dual algorithm:
\begin{subequations}\label{centr_pd_opf}
\begin{align}
x[k+1] &= \mathcal{P}_{\mathcal{X}}\left(x[k] - s\frac{\partial L[k]}{\partial x}\right),\label{x_update}\\
\omega[k+1] &= \mathcal{P}_{\Omega}\left(\omega[k] - s\frac{\partial L[k]}{\partial \omega}\right),\label{omega_update}\\
\gamma[k+1] &= \gamma[k] + s\frac{\partial L[k]}{\partial \gamma},\\
\tau[k+1] &= \Big[\tau[k] + 2s\frac{\partial L[k]}{\partial \tau}\Big]_+,\label{tau_update}
\end{align}
\end{subequations}
where $L[k] \coloneqq L(x[k],\omega[k],\gamma[k],\tau[k])$, $\mathcal{P}_{\mathcal{X}}(\cdot)$ denotes the projection onto the set \[\mathcal{X} \coloneqq \big\{x: w \in[w^{\min},w^{\max}], w \in\{g^{(p)},g^{(q)},v\}\big\},\] $\mathcal{P}_{\Omega}(\cdot)$ denotes the projection onto the set \[\Omega \coloneqq \big\{\omega: w \in[w^{\min},w^{\max}], w \in\{p,q,\varepsilon\},\ell\in[0,\ell^{\max}]\big\},\]
%\mathcal{X} &\coloneqq \{x: g^{(p)}\in[g^{(p)\min},g^{(p)\max}], g^{(q)}\in[g^{(q)\min},g^{(q)\max}],\\&\qquad v\in[v^{\min},v^{\max}], \ell\in[0,\ell^{\max}], p\in[p^{\min},p^{\max}],\\&\qquad q\in[q^{\min},q^{\max}] \},
%\end{align*}
with $p^{\max} \coloneqq (\ell^{\max})^{1/2}\circ (v^{\max})^{1/2}$, $q^{\max} \coloneqq (\ell^{\max})^{1/2}\circ (v^{\max})^{1/2}$, $\varepsilon^{\max}\coloneqq v^{\max}-v^{\min}$, $p^{\min} \coloneqq -p^{\max}$, $q^{\min} \coloneqq -q^{\max}$, $\varepsilon^{\min}\coloneqq-\varepsilon^{\max}$, and $[\cdot]_{+}$ denotes the projection onto the interval~$[0,+\infty)$. Notice that the~$\tau$-update \eqref{tau_update} uses $2\frac{\partial L[k]}{\partial \tau}$ instead of simply using $\frac{\partial L[k]}{\partial \tau}$; this subtle change (to be clarified later when we present the convergence analysis) is due to the nonlinearity of the constraint \eqref{DistFlow5}. Let $x^* \coloneqq [g^{(p)*},g^{(q)*},v^*]^\top$, $\omega^* \coloneqq [p^*,q^*,\varepsilon^*,\ell^*]^\top$, $\gamma^* \coloneqq [\lambda^*,\mu^*,\nu^*,\eta^*]^\top$, and $\tau^*$ denote the equilibrium of \eqref{centr_pd_opf}.

In the proposed distributed version of \eqref{centr_pd_opf}, each node $i$ estimates the optimal values of only local primal and dual variables, denoted by
$\mathbf{x}^{(i)*} \coloneqq [x^{(i)*\top},\omega^{(i)*\top},\gamma^{(i)*\top},\tau_i^*]^\top$, with
\begin{align*}
x^{(i)*} &\coloneqq [g^{(p)*}_i,g^{(q)*}_i,v^*_i]^\top,\\
\omega^{(i)*} &\coloneqq \big[[\{p_{ij}^*\}_{(i,j)\in\mathcal{E}_p}],[\{p_{li}^*\}_{(l,i)\in\mathcal{E}_p}],\nonumber\\
&\qquad[\{q_{ij}^*\}_{(i,j)\in\mathcal{E}_p}],[\{q_{li}^*\}_{(l,i)\in\mathcal{E}_p}],\nonumber\\
&\qquad[\{\varepsilon_{ij}^*\}_{(i,j)\in\mathcal{E}_p}],[\{\varepsilon_{li}^*\}_{(l,i)\in\mathcal{E}_p}],\nonumber\\
&\qquad[\{\ell_{ij}^*\}_{(i,j)\in\mathcal{E}_p}],[\{\ell_{li}^*\}_{(l,i)\in\mathcal{E}_p}]\big]^\top,\\
\gamma^{(i)*} &\coloneqq \big[\lambda_i^*,\mu_i^*,[\{\nu_{ij}^*\}_{(i,j)\in\mathcal{E}_p}],[\{\nu_{li}^*\}_{(l,i)\in\mathcal{E}_p}],\nonumber\\
&\qquad[\{\eta_{ij}^*\}_{(i,j)\in\mathcal{E}_p}],[\{\eta_{li}^*\}_{(l,i)\in\mathcal{E}_p}]\big]^\top.
\end{align*}
Let $\mathbf{x}^{(i)}[k] \coloneqq [x^{(i)}[k]^\top,\omega^{(i)}[k]^\top,\gamma^{(i)}[k]^\top,\tau_i[k]]^\top$ denote the vector of estimates of $\mathbf{x}^{(i)*}$ maintained at node~$i$ at time instant~$k$, where
\begin{align*}
x^{(i)}[k] &\coloneqq [g^{(p)}_i[k],g^{(q)}_i[k],v_i[k]]^\top,\\
\omega^{(i)}[k] &\coloneqq \big[[\{p_{ij}[k]\}_{(i,j)\in\mathcal{E}_p}],[\{p_{li}[k]\}_{(l,i)\in\mathcal{E}_p}],\nonumber\\
&\qquad[\{q_{ij}[k]\}_{(i,j)\in\mathcal{E}_p}],[\{q_{li}[k]\}_{(l,i)\in\mathcal{E}_p}],\nonumber\\
&\qquad[\{\varepsilon_{ij}[k]\}_{(i,j)\in\mathcal{E}_p}],[\{\varepsilon_{li}[k]\}_{(l,i)\in\mathcal{E}_p}],\nonumber\\
&\qquad[\{\ell_{ij}[k]\}_{(i,j)\in\mathcal{E}_p}],[\{\ell_{li}[k]\}_{(l,i)\in\mathcal{E}_p}]\big]^\top,\\
\gamma^{(i)}[k] &\coloneqq \big[\lambda_i[k],\mu_i[k],[\{\nu_{ij}[k]\}_{(i,j)\in\mathcal{E}_p}],[\{\nu_{li}[k]\}_{(l,i)\in\mathcal{E}_p}],\nonumber\\
&\qquad[\{\eta_{ij}[k]\}_{(i,j)\in\mathcal{E}_p}],[\{\eta_{li}[k]\}_{(l,i)\in\mathcal{E}_p}]\big]^\top.
\end{align*}
%$x^{(i)*}\coloneqq[g_i^{(p)*}, g_i^{(q)*}, v_i, \ell_{ij}^*, p_{ij}^*, q_{ij}^*, \varepsilon_{ij}^*, \nu_{ij}^*, \eta_{ij}^*, \forall j\in\mathcal{D}_i, \lambda_i^*, \mu_i^*, \ell_{li}^*,\allowbreak p_{li}^*, q_{li}^*, \varepsilon_{li}^*, \nu_{li}^*, \eta_{li}^*, \tau_{li}^*, \forall l \in \mathcal{U}_i]$. 
%In the proposed distributed version of \eqref{centr_pd_opf}, every node $i$ estimates the optimal values of only local quantities, i.e., $g^{(p)*}_i$, $g^{(q)*}_i$, $v_i^*$, $\ell_{ij}^*$, $p_{ij}^*$, $q_{ij}^*$, $\lambda_i^*$, $\mu_i^*$, $\nu_{ij}^*$, $\varepsilon_{ij}^*$, $(i,j)\in\mathcal{E}_p$, and $\ell_{ji}^*$, $p_{ji}^*$, $q_{ji}^*$, $\nu_{ji}^*$, $\varepsilon_{ji}^*$, $\tau_{ji}^*$, $(j,i)\in\mathcal{E}_p$. Let $\hat{\ell}_{ij}[k]$, $\hat{p}_{ij}[k]$, $\hat{q}_{ij}[k]$, $\hat{\nu}_{ij}[k]$, $\hat{\varepsilon}_{ij}[k]$ denote the estimates of $\ell_{ij}^*$, $p_{ij}^*$, $q_{ij}^*$, $\nu_{ij}^*$, $\varepsilon_{ji}^*$, respectively, $(i,j)\in\mathcal{E}_p$, $\forall j\in\mathcal{U}_i$. Let $\breve{\ell}_{li}[k]$, $\breve{p}_{li}[k]$, $\breve{q}_{li}[k]$, $\breve{\nu}_{li}[k]$, $\breve{\varepsilon}_{li}[k]$, $\tau_{li}[k]$ denote the estimates of $\ell_{li}^*$, $p_{li}^*$, $q_{li}^*$, $\nu_{li}^*$, $\varepsilon_{li}^*$, and $\tau_{li}^*$, respectively, $(l,i)\in\mathcal{E}_p$, $\forall l\in\mathcal{D}_i$. 
Then, node~$i$ performs updates using the following local Lagrangian:
\begin{align*}
&L^{(i)}(\mathbf{x}^{(i)})  = f_i(g_i^{(p)}) +\rho(v_i-1)^2 + \lambda_i b_i^{(p)} + \mu_i b_i^{(q)}\\
&\mbox{ } + \rho_1 (b_i^{(p)})^2+ \rho_2(b_i^{(q)})^2+ \sum_{{(i,j)\in\mathcal{E}_p}}\frac{\rho}{2}\hat{\ell}_{ij}^2 + \sum_{(l,i)\in\mathcal{E}_p}\frac{\rho}{2}\breve{\ell}_{li}^2\\
&\mbox{ }+ \sum_{(i,j)\in\mathcal{E}_p}\frac{1}{2}\hat{\eta}_{ij}(\hat{\varepsilon}_{ij} -2 v_i) + \sum_{(l,i)\in\mathcal{E}_p}\frac{1}{2}\breve{\eta}_{li}(\breve{\varepsilon}_{li} + 2v_i)\\
&\mbox{ }+ \sum_{\mathclap{(i,j)\in\mathcal{E}_p}}\Big(\frac{\hat{\nu}_{ij}}{2} \hat{b}_{ij}^{(v)} + \frac{\rho_3}{2}(\hat{b}_{ij}^{(v)})^2\Big) + \sum_{\mathclap{(l,i)\in\mathcal{E}_p}}\Big(\frac{\breve{\nu}_{li}}{2} \breve{b}_{li}^{(v)} + \frac{\rho_3}{2}(\breve{b}_{li}^{(v)})^2\Big)\\
%&\quad+ \sum_{(i,j)\in\mathcal{E}_p}\hat{\nu}_{ij}(\hat{\varepsilon}_{ij} - 2r_{ij}\hat{p}_{ij}-2x_{ij}\hat{q}_{ij}+(r_{ij}^2+x_{ij}^2)\hat{\ell}_{ij})\\
%&\quad+ \sum_{(l,i)\in\mathcal{E}_p}\breve{\nu}_{li}(\breve{\varepsilon}_{li} - 2r_{li}\breve{p}_{li}-2x_{li}\breve{q}_{li}+(r_{li}^2+x_{li}^2)\breve{\ell}_{li})\\
&\mbox{ }+\sum_{(i,j)\in\mathcal{E}_p}\tau_{ij}(\hat{p}_{ij}^2+\hat{q}_{ij}^2- v_i\hat{\ell}_{ij}),
\end{align*}
%which was obtained from the Lagrangian $L(x,\gamma,\tau)$ by collecting all terms that are local to node $i$, 
%where $x^{(i)}\coloneqq[g_i^{(p)}, g_i^{(q)}, v_i, \hat{\ell}_{ij}, \hat{p}_{ij}, \hat{q}_{ij}, \hat{\varepsilon}_{ij}, \hat{\nu}_{ij}, \hat{\eta}_{ij},\forall j\in\mathcal{D}_i,\lambda_i, \mu_i,\allowbreak \breve{\ell}_{li}, \breve{p}_{li}, \breve{q}_{li}, \breve{\varepsilon}_{li}, \breve{\nu}_{li}, \breve{\eta}_{li}, \tau_{li}, \forall l \in \mathcal{U}_i]$, and
where
\begin{align*}
\hat{b}_{ij}^{(v)} &\coloneqq \hat{\varepsilon}_{ij} - 2r_{ij}\hat{p}_{ij}-2x_{ij}\hat{q}_{ij}+(r_{ij}^2+x_{ij}^2)\hat{\ell}_{ij},(i,j)\in\mathcal{E}_p,\\
\breve{b}_{li}^{(v)} &\coloneqq \breve{\varepsilon}_{li} - 2r_{li}\breve{p}_{li}-2x_{li}\breve{q}_{li}+(r_{li}^2+x_{li}^2)\breve{\ell}_{li},(l,i)\in\mathcal{E}_p,\\
b^{(p)}_i &\coloneqq g^{(p)}_i - l^{(p)}_i - \sum_{(i,j)\in\mathcal{E}_p}\hat{p}_{ij} + \sum_{(l,i)\in\mathcal{E}_p}\breve{p}_{li} - \sum_{(l,i)\in\mathcal{E}_p}r_{li}\breve{\ell}_{li},\\
b^{(q)}_i &\coloneqq g^{(q)}_i - l^{(q)}_i - \sum_{(i,j)\in\mathcal{E}_p}\hat{q}_{ij} + \sum_{(l,i)\in\mathcal{E}_p}\breve{q}_{li} - \sum_{(l,i)\in\mathcal{E}_p}x_{li}\breve{\ell}_{li}.
\end{align*}
Note that local Lagrangian~$L^{(i)}(\mathbf{x}^{(i)})$ is obtained from the Lagrangian~$L(x,\omega,\gamma,\tau)$ by collecting all terms that are local to node~$i$ such that those terms that are also local to neighboring nodes are decomposed into equal parts as shown below:
\begin{figure}[H]\vspace{-10pt}
    \centering 
	\includegraphics[trim=0.2cm 0cm 0cm 0cm, clip=true, scale=0.84]{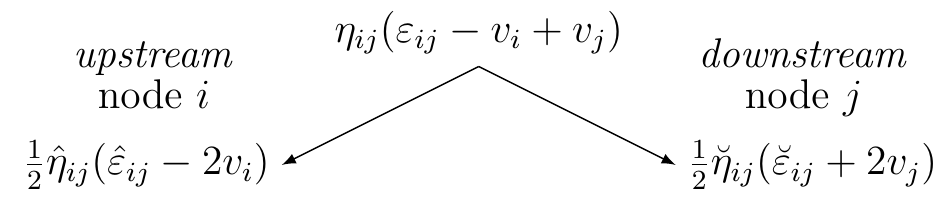} \vspace{-10pt}
 \end{figure}
We use exactly the same ideas presented in Section~\ref{sec:sced}, and propose the following distributed version of \eqref{centr_pd_opf}:
\begin{subequations}\label{pd_opf_un1}
\begin{align}
\chi[k+1] &= \Big[\chi[k] - s\frac{\partial L^{(i)}[k]}{\partial \chi}\Big]_{\chi^{\min}}^{\chi^{\max}},\\
\chi &\in \{g^{(p)}_i,g^{(q)}_i,v_i\},\nonumber\\
\psi[k+1] &= \psi[k] + s\frac{\partial L^{(i)}[k]}{\partial \psi},\\
\psi &\in \{\lambda_i,\mu_i\},\nonumber\\
\hat{w}_{ij}[k+1] &= \Big[(1-a_{ij}[k])\hat{w}_{ij}[k]+ a_{ij}[k]\breve{w}_{ij}[k]\nonumber\\
&\quad  - s\hat{y}^{(w)}_{ij}[k]\Big]_{w_{ij}^{\min}}^{w_{ij}^{\max}},\\
\breve{w}_{li}[k+1] &= \Big[(1-a_{li}[k])\breve{w}_{li}[k]+ a_{li}[k]\hat{w}_{li}[k]\nonumber\\
&\quad - s\breve{y}^{(w)}_{li}\Big]_{w_{li}^{\min}}^{w_{li}^{\max}},\\
w &\in \{\ell, p, q,\varepsilon\},\nonumber\\
\hat{d}_{ij}[k+1] &= (1-a_{ij}[k])\hat{d}_{ij}[k]+ a_{ij}[k]\breve{d}_{ij}[k] \nonumber\\&\quad + s\hat{y}^{(d)}_{ij}[k],\\
\breve{d}_{li}[k+1] &= (1-a_{li}[k])\breve{d}_{li}[k]+ a_{li}[k]\hat{d}_{li}[k]\nonumber\\&\quad + s\breve{y}^{(d)}_{li}[k],\\
d &\in \{\nu,\eta\},\nonumber\\
\tau_{li}[k+1] &= \Big[\tau_{li}[k] + 2s\frac{\partial L^{(i)}[k]}{\partial \tau_{li}}\Big]_+,\\
(i,j)&\in\mathcal{E}_p,(l,i)\in\mathcal{E}_p,\nonumber
\end{align}
\end{subequations}
where $L^{(i)}[k]\coloneqq L^{(i)}(\mathbf{x}^{(i)}[k])$, and 
\begin{align*}
a_{ij}[k]=\left\{\begin{matrix*}[l]0.5 & \mbox{if }\{i,j\} \in \mathcal{E}_c[k],\\0 & \mbox{otherwise.}\end{matrix*}\right.
\end{align*}
The gradients~$\hat{y}^{(w)}[k]\coloneqq [\{\hat{y}_{ij}^{(w)}[k]\}_{(i,j)\in\mathcal{E}_p}]$ and $\breve{y}^{(w)}[k]\coloneqq [\{\breve{y}_{li}^{(w)}[k]\}_{(l,i)\in\mathcal{E}_p}]$, $w \in \{\ell, p, q,\varepsilon,\nu,\eta\}$, are updated as follows:
\begin{subequations}\label{pd_opf_un2}
\begin{align}
&\hat{y}^{(w)}_{ij}[k+1] = (1-a_{ij}[k]) \hat{y}^{(w)}_{ij}[k] + a_{ij}[k]\breve{y}^{(w)}_{ij}[k]\nonumber\\
&\quad + 2\left(\frac{\partial L^{(i)}[k+1]}{\partial \hat{w}_{ij}} - \frac{\partial L^{(i)}[k]}{\partial \hat{w}_{ij}}\right),(i,j)\in\mathcal{E}_p, \\
&\breve{y}^{(w)}_{li}[k+1] = (1-a_{li}[k]) \breve{y}^{(w)}_{li}[k] + a_{li}[k]\hat{y}^{(w)}_{li}[k]\nonumber\\
&\quad + 2\left(\frac{\partial L^{(i)}[k+1]}{\partial \breve{w}_{li}} - \frac{\partial L^{(i)}[k]}{\partial \breve{w}_{li}}\right),(l,i)\in\mathcal{E}_p.
\end{align}
\end{subequations}
We initialize \eqref{pd_opf_un2} as follows:
\begin{subequations}\label{init}
\begin{align}
\hat{y}^{(w)}_{ij}[0] &= 2\frac{\partial L^{(i)}[0]}{\partial \hat{w}_{ij}}, (i,j)\in\mathcal{E}_p, \\
\breve{y}^{(w)}_{li}[0] &= 2\frac{\partial L^{(i)}[0]}{\partial \breve{w}_{li}}, (l,i)\in\mathcal{E}_p. 
\end{align}
\end{subequations}
We noticed from the numerical simulations that if the initial voltage magnitudes are set to $1.0$~per unit, $y^{(\eta)}[0]$ must be initialized differently to achieve a better performance. If \eqref{init} is used, then,
\begin{align}
\begin{bmatrix} \hat{y}^{(\eta)}[0]\\\breve{y}^{(\eta)}[0]\end{bmatrix} &= \begin{bmatrix}\hat{\varepsilon}[0]\\\breve{\varepsilon}[0]\end{bmatrix} - 2\begin{bmatrix}M_0^\top v[0]\\N_0^\top v[0]\end{bmatrix}.\label{y_eta}
\end{align}
If $v_i[0] = 1$, $i \in \mathcal{V}_p$, then, the second term on the right-hand side of \eqref{y_eta} does not have any effect on the average estimate, $\frac{1}{2}(\hat{y}^{(\eta)}[0]+\breve{y}^{(\eta)}[0])$, since $M_0^\top v[0]+N_0^\top v[0] = Mv[0] = 0$. In view of this observation and the fact that the second term in \eqref{y_eta} can be significantly larger than the first term, it is better to neglect it during the initialization.

\subsection{Feedback Interconnection Representation of the Distributed Primal-Dual Algorithm}\label{subsec:feedback_opf}
Similar to the feedback representation given in Section~\ref{subsec:feedback_sced}, we show that~\eqref{pd_opf_un1}--\eqref{pd_opf_un2} can also be represented as a feedback interconnection of a nominal system, denoted by $\mathcal{H}^{opf}_1$, and a disturbance system, denoted by $\mathcal{H}^{opf}_2$, which allows us to utilize the small-gain theorem for convergence analysis purposes. To this end, let 
\begin{align*}
\hat{w}[k]&\coloneqq [\{\hat{w}_{ij}[k]\}_{(i,j)\in\mathcal{E}_p}], \breve{w}[k]\coloneqq [\{\breve{w}_{li}[k]\}_{(l,i)\in\mathcal{E}_p}],\\
\overline{w}[k]&\coloneqq \frac{1}{2}(\hat{w}[k]+\breve{w}[k]), w \in \{\ell,p, q, \varepsilon, \nu, \eta\},\\
\hat{\omega}[k] &\coloneqq [\hat{p}[k]^\top,\hat{q}[k]^\top,\hat{\varepsilon}[k]^\top,\hat{\ell}[k]^\top]^\top,\\
\breve{\omega}[k] &\coloneqq [\breve{p}[k]^\top,\breve{q}[k]^\top,\breve{\varepsilon}[k]^\top,\breve{\ell}[k]^\top]^\top,\\
\hat{\gamma}[k] &\coloneqq [\lambda[k]^\top,\mu[k]^\top,\hat{\nu}[k]^\top,\hat{\eta}[k]^\top]^\top,\\
\breve{\gamma}[k] &\coloneqq [\lambda[k]^\top,\mu[k]^\top,\breve{\nu}[k]^\top,\breve{\eta}[k]^\top]^\top,\\
\hat{y}[k] &\coloneqq \big[\hat{y}^{(\ell)}[k]^\top,\hat{y}^{(p)}[k]^\top,\hat{y}^{(q)}[k]^\top,\hat{y}^{(\varepsilon)}[k]^\top,\hat{y}^{(\nu)}[k]^\top,\\&\qquad\hat{y}^{(\eta)}[k]^\top\big]^\top,\\
\breve{y}[k] &\coloneqq \big[\breve{y}^{(\ell)}[k]^\top,\breve{y}^{(p)}[k]^\top,\breve{y}^{(q)}[k]^\top,\breve{y}^{(\varepsilon)}[k]^\top,\breve{y}^{(\nu)}[k]^\top,\\&\qquad\breve{y}^{(\eta)}[k]^\top\big]^\top,\\
%\overline{\omega} &\coloneqq [\overline{p}^\top,\overline{q}^\top,\overline{\varepsilon}^\top,\overline{\ell}^\top]^\top, \overline{\gamma} \coloneqq [\lambda^\top,\mu^\top,\overline{\nu}^\top,\overline{\eta}^\top]^\top,\\
\overline{\omega}[k] &\coloneqq \frac{1}{2}(\hat{\omega}[k]+\breve{\omega}[k]),
\overline{\gamma}[k] \coloneqq \frac{1}{2}(\hat{\gamma}[k]+\breve{\gamma}[k]),\\
\overline{y}[k] &\coloneqq \frac{1}{2}(\hat{y}[k]+\breve{y}[k]),
\end{align*}
By using \eqref{pd_opf_un1}, we compactly write the iterations for $x$, $\overline{\omega}$, $\overline{\gamma}$, and $\tau$, which constitute the nominal system, $\mathcal{H}^{opf}_1$, given by:
\begin{subequations}\label{H1_opf_un}
\begin{align}
{\mathcal{H}_1^{opf}:}\mbox{ }x[k+1] &= \mathcal{P}_{\mathcal{X}}\Big(x[k] - s\frac{\partial \overline{L}[k]}{\partial x}+e_x[k]\Big),\label{H1_opf_x}\\
\overline{\omega}[k+1] &= \frac{1}{2}\mathcal{P}_{\Omega}\Big(\overline{\omega}[k] - s\frac{\partial \overline{L}[k]}{\partial \omega}+e_\omega[k]\Big)\nonumber\\&\quad+
\frac{1}{2}\mathcal{P}_{\Omega}\Big(\overline{\omega}[k] - s\frac{\partial \overline{L}[k]}{\partial \omega}-e_\omega[k]\Big),\\
\overline{\gamma}[k+1] &= \overline{\gamma}[k] + s\frac{\partial \overline{L}[k]}{\partial \gamma} + e_\gamma[k],\\
\tau[k+1] &= \Big[\tau[k] + 2s\frac{\partial \overline{L}[k]}{\partial \tau} + e_{\tau}[k]\Big]_+,
\end{align}
\end{subequations}
where $\overline{L}[k] \coloneqq L(x,\overline{\omega}[k],\overline{\gamma}[k],\tau[k])$, $e[k]\coloneqq [e_x[k]^\top,e_\omega[k]^\top,e_{\gamma}[k]^\top,e_{\tau}[k]^\top]^\top$ is given by 
\begin{align*}
e[k]=A\begin{bmatrix}\hat{\omega}[k]-\overline{\omega}[k]\\\hat{\gamma}[k]-\overline{\gamma}[k]\\\hat{y}[k]-\overline{y}[k]\end{bmatrix},
\end{align*}
for some constant matrix~$A$. Notice that $e[k]$ results from $(\hat{\omega}[k],\breve{\omega}[k])$, $(\hat{\gamma}[k],\breve{\gamma}[k])$, and $(\hat{y}[k],\breve{y}[k])$ deviating from their respective average, $\overline{\omega}[k]$, $\overline{\gamma}[k]$, and $\overline{y}[k]$; without $e[k]$, the nominal system~$\mathcal{H}^{opf}_1$ has exactly the same form as~\eqref{centr_pd_opf}.
Now, we define the disturbance system, $\mathcal{H}^{opf}_2$, as follows: 
\begin{subequations}\label{H2_opf_un}
\begin{align}
%{\mathcal{H}_2^{opf}:}\mbox{ }\\
\hat{w}_{ij}[k+1] &= \Big[(1-a_{ij}[k])\hat{w}_{ij}[k]+ a_{ij}[k]\breve{w}_{ij}[k]\nonumber\\
&\quad  - s\hat{y}^{(w)}_{ij}[k]\Big]_{w_{ij}^{\min}}^{w_{ij}^{\max}}, \\
\breve{w}_{li}[k+1] &= \Big[(1-a_{li}[k])\breve{w}_{li}[k]+ a_{li}[k]\hat{w}_{li}[k]\nonumber\\
&\quad - s\breve{y}^{(w)}_{li}\Big]_{w_{li}^{\min}}^{w_{li}^{\max}},\\
w &\in \{\ell, p, q,\varepsilon\},(i,j)\in\mathcal{E}_p,(l,i)\in\mathcal{E}_p,\nonumber\\
\hat{d}_{ij}[k+1] &= (1-a_{ij}[k])\hat{d}_{ij}[k]+ a_{ij}[k]\breve{d}_{ij}[k] \nonumber\\&\quad + s\hat{y}^{(d)}_{ij}[k],\\
\breve{d}_{li}[k+1] &= (1-a_{li}[k])\breve{d}_{li}[k]+ a_{li}[k]\hat{d}_{li}[k]\nonumber\\&\quad + s\breve{y}^{(d)}_{li}[k],\\
d&\in \{\nu,\eta\},\nonumber\\
\hat{y}^{(\psi)}_{ij}[k+1] &= (1-a_{ij}[k]) \hat{y}^{(\psi)}_{ij}[k] + a_{ij}[k]\breve{y}^{(\psi)}_{ij}[k]\nonumber\\
&\quad + 2\left(\frac{\partial L^{(i)}[k+1]}{\partial \hat{w}_{ij}} - \frac{\partial L^{(i)}[k]}{\partial \hat{w}_{ij}}\right), \\
\breve{y}^{(\psi)}_{li}[k+1] &= (1-a_{li}[k]) \breve{y}^{(\psi)}_{li}[k] + a_{li}[k]\hat{y}^{(\psi)}_{li}[k]\nonumber\\
&\quad + 2\left(\frac{\partial L^{(i)}[k+1]}{\partial \breve{w}_{li}} - \frac{\partial L^{(i)}[k]}{\partial \breve{w}_{li}}\right),\\
\psi& \in \{\ell, p, q,\varepsilon,\nu,\eta\},\nonumber\\
e[k]&=A\begin{bmatrix}\hat{\omega}[k]-\overline{\omega}[k]\\\hat{\gamma}[k]-\overline{\gamma}[k]\\\hat{y}[k]-\overline{y}[k]\end{bmatrix},
\end{align}
\end{subequations}
Then, as depicted in Fig.~\ref{fig:feedback_opf}, algorithm~\eqref{pd_opf_un1}--\eqref{pd_opf_un2} can be viewed as a feedback interconnection of $\mathcal{H}^{opf}_1$ and $\mathcal{H}^{opf}_2$.
%, where $(x^*,\omega^*,\gamma^*,\tau^*)$ is the equilibrium of~\eqref{H1_opf_un} when $e[k]\equiv 0$, for all $k\geq0$.
\begin{figure}
    \centering 
	\includegraphics[trim=0cm 0cm 0cm 0cm, clip=true, scale=0.6]{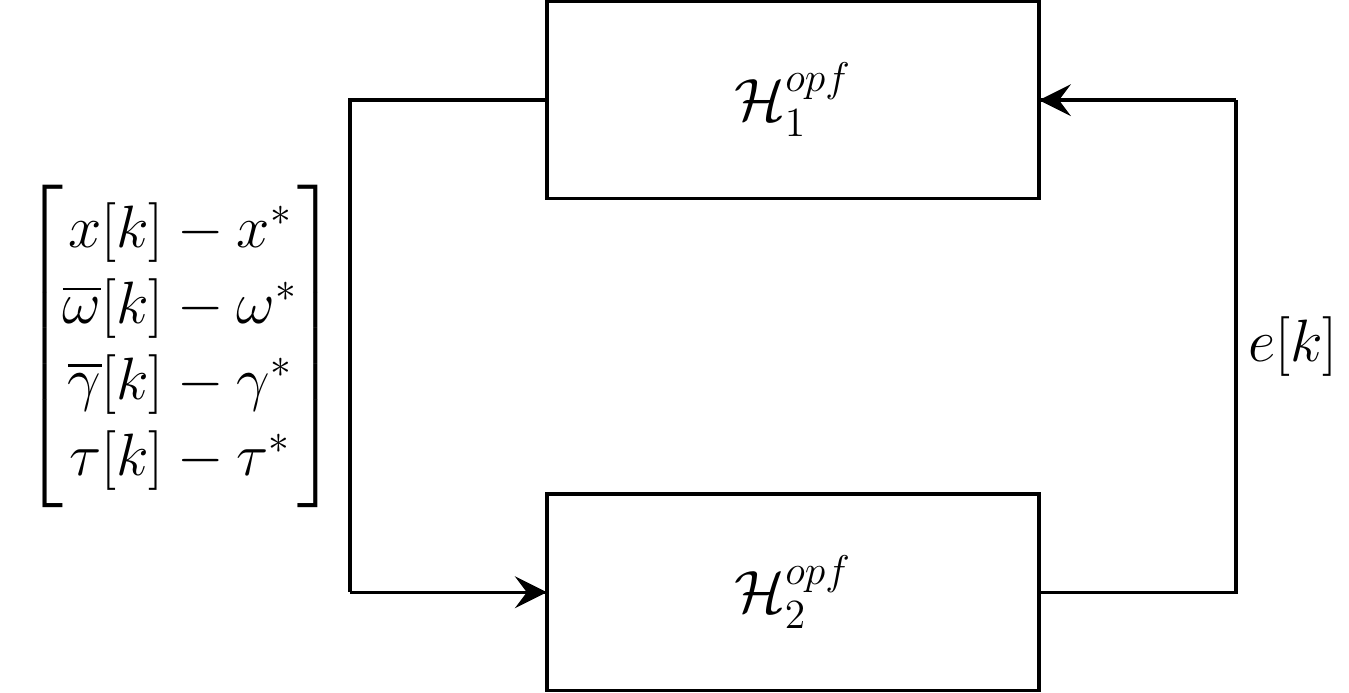} 
    \caption{Algorithm~\eqref{pd_opf_un1}--\eqref{pd_opf_un2} as a feedback system.}
    \vspace{-5pt} 
    \label{fig:feedback_opf}
\end{figure}

%To enable the use of the small-gain theorem, we later prove the following key results:
%\begin{itemize}
%\item[{\textbf{R1.}}] $\|z\|_2^{a,K}\leq \alpha_1\|e\|_2^{a,K}+\beta_1$ for some constant $\alpha_1$ and $\beta_1$,
%\item[{\textbf{R2.}}] $\|e\|_2^{a,K}\leq s\alpha_2\|z\|_2^{a,K}+\beta_2$ for some constant $\alpha_2$ and $\beta_2$,
%\end{itemize}
%where $z[k] \coloneqq [(\overline{x}[k] - x^*)^\top,(\overline{\gamma}[k] - \gamma^*)^\top,(\tau[k] - \tau^*)^\top]^\top$ denotes the convergence error, $e[k]\coloneqq [e_x[k]^\top,e_\gamma[k]^\top,e_\tau[k]^\top]^\top$ denotes the disturbance, $(x^*,\gamma^*,\tau^*)$ denotes the equilibrium of the nominal system, and $x^*$ is the solution of rSOCP.
%For small enough $s$, the gain of the feedback connection, $s\alpha_1\alpha_2$, can be made strictly smaller than $1$; then, it follows from the small-gain theorem that $\|z\|_2^{a,K}$ is bounded for all $K>0$, and $z[k]$ converges to zero geometrically fast.
%
%Next, we present the convergence results for the distributed algorithm~\eqref{pd_opf_un1} by using this control-theoretic interpretation. 
%%The disturbance $[e_x,e_{\gamma},e_{\tau}]$ is caused by the presence of inactive communication links.
%%In fact, if communication network is time-invariant, then, the disturbance disappears. 
%%Without the disturbance $[e_x,e_{\gamma},e_{\tau}]$, \eqref{H1_opf_un} has exactly the same form as \eqref{centr_pd_opf}, the nominal system. In what follows, we present the convergence results for algorithm \eqref{pd_opf_un1}. 

\subsection{Convergence Analysis}\label{subsec:conv_analysis_opf_un}
Representing algorithm~\eqref{pd_opf_un1}--\eqref{pd_opf_un2} in a feedback form facilitates our analysis relying on the small-gain theorem.
Much of the analysis is dedicated to establishing that $\mathcal{H}^{opf}_1$ and $\mathcal{H}^{opf}_2$ are finite-gain stable. Then, the small-gain theorem is applied in a straightforward manner to show the convergence of~\eqref{pd_opf_un1}--\eqref{pd_opf_un2}.

In the next result, we establish that $\mathcal{H}^{opf}_1$ is finite-gain stable.
%%%%%%%%%%%%%%%%%%%%%%%%%%%
%%%%%%%% H1_opf_un is finite-gain stable %%%%%%%%
\begin{proposition}\label{prop:H1_opf_un}
Let Assumption~\ref{objective_assumption} hold.
Then, under~\eqref{H1_opf_un}, we have that
\begin{align}
\textnormal{\textbf{R1. }} \|z\|_2^{a,K}\leq \alpha_1\|e\|_2^{a,K}+\beta_1,
\label{z_to_e_opf_un}
\end{align}
for some positive~$\alpha_1$ and $\beta_1$, $a \in (0,1)$, and sufficiently small~$s>0$,
where
\[z[k] \coloneqq \begin{bmatrix}x[k] - x^*\\\overline{\omega}[k] - \omega^*\\\overline{\gamma}[k] - \gamma^*\\ \tau[k] - \tau^*\\ \end{bmatrix}.\]
\end{proposition}
\begin{proof}
Letting $L^* \coloneqq L(x^*,\omega^*,\gamma^*,\tau^*)$, 
\begin{align*}
G[k] &\coloneqq \begin{bmatrix}
x[k] - s\frac{\partial \overline{L}[k]}{\partial x}\\[2pt]
\overline{\omega}[k] - s\frac{\partial \overline{L}[k]}{\partial \omega}\\[2pt]
\overline{\gamma}[k] + s\frac{\partial \overline{L}[k]}{\partial\gamma}\\[2pt]
\tau[k] + 2s\frac{\partial \overline{L}[k]}{\partial \tau}
\end{bmatrix},\mbox{ and }
G^* \coloneqq \begin{bmatrix}
x^* - s\frac{\partial L^*}{\partial x}\\[2pt]
\omega^* - s\frac{\partial L^*}{\partial \omega}\\[2pt]
\gamma^* + s\frac{\partial L^*}{\partial \gamma}\\[2pt]
\tau^* + 2s\frac{\partial L^*}{\partial \tau}
\end{bmatrix},
\end{align*}
we establish the following result.
\begin{lemma}\label{lem:projection}
\begin{align}
\|z[k+1]\|&\leq \|G[k]-G^*\| + \|e[k]\|,
\label{projection}
\end{align}
where $\|\cdot\|$ is a vector norm.
\end{lemma}
\begin{proof}
We note that the following relationship holds:
\begin{align}
x^* = \mathcal{P}_{\mathcal{X}}\Big(x^* - s\frac{\partial L^*}{\partial x}\Big).\label{lem_proj_eq1}
\end{align}
Then, by applying the triangle inequality on multiple occasions and the Projection Theorem \cite[Proposition~2.1.3]{NonlinearProgramming}, we have that
\begin{subequations}\label{lem_projection}
\begin{align}
\|x[k+1]&-x^*\|=\Big\|\frac{1}{2}\mathcal{P}_{\mathcal{X}}\Big(x[k] - s\frac{\partial \overline{L}[k]}{\partial x}+e_x[k]\Big)\nonumber\\
&\quad+\frac{1}{2}\mathcal{P}_{\mathcal{X}}\Big(x[k] - s\frac{\partial \overline{L}[k]}{\partial x}-e_x[k]\Big) - x^*\Big\|\label{lem_proj_eq2}\\
&\leq \frac{1}{2}\Big\|\mathcal{P}_{\mathcal{X}}\Big(x[k] - s\frac{\partial \overline{L}[k]}{\partial x}+e_x[k]\Big)- x^*\Big\|\nonumber\\
&\quad+\frac{1}{2}\Big\|\mathcal{P}_{\mathcal{X}}\Big(x[k] - s\frac{\partial \overline{L}[k]}{\partial x}-e_x[k]\Big) - x^*\Big\|\label{lem_proj_eq3}\\
&= \frac{1}{2}\Big\|\mathcal{P}_{\mathcal{X}}\Big(x[k] - s\frac{\partial \overline{L}[k]}{\partial x}+e_x[k]\Big)\nonumber\\&\quad- \mathcal{P}_{\mathcal{X}}\Big(x^* - s\frac{\partial L^*}{\partial x}\Big)\Big\|\nonumber\\
&\quad+\frac{1}{2}\Big\|\mathcal{P}_{\mathcal{X}}\Big(x[k] - s\frac{\partial \overline{L}[k]}{\partial x}-e_x[k]\Big) - \nonumber\\&\quad- \mathcal{P}_{\mathcal{X}}\Big(x^* - s\frac{\partial L^*}{\partial x}\Big)\Big\|\label{lem_proj_eq4}\\
&\leq \frac{1}{2}\Big\|x[k] - s\frac{\partial \overline{L}[k]}{\partial x} - x^* + s\frac{\partial L^*}{\partial x}\nonumber\\
&\quad+e_x[k]\Big\|+\frac{1}{2}\Big\|x[k] - s\frac{\partial \overline{L}[k]}{\partial x} - x^* \nonumber\\&\quad+ s\frac{\partial L^*}{\partial x}-e_x[k]\Big\|\label{lem_proj_eq5}\\
&\leq \Big\|x[k] - s\frac{\partial \overline{L}[k]}{\partial x} - x^* + s\frac{\partial L^*}{\partial x}\Big\|+\|e_x[k]\|,
\label{lem_proj_eq6}
\end{align}
where we used \eqref{H1_opf_x} to obtain \eqref{lem_proj_eq2}, applied the triangle inequality to obtain \eqref{lem_proj_eq3} and \eqref{lem_proj_eq6}, used \eqref{lem_proj_eq1} to obtain \eqref{lem_proj_eq4}, and the projection theorem to obtain \eqref{lem_proj_eq5}.
Similarly, we find that
\begin{align}
\|\overline{\omega}[k+1]-\omega^*\|&\leq \Big\|\overline{\omega}[k] - s\frac{\partial \overline{L}[k]}{\partial \omega} - \omega^* + s\frac{\partial L^*}{\partial \omega}\Big\|\nonumber\\&\quad+\|e_\omega[k]\|,\\
\|\overline{\gamma}[k+1]-\gamma^*\|&\leq \Big\|\overline{\gamma}[k] - s\frac{\partial \overline{L}[k]}{\partial \omega} - \omega^* + s\frac{\partial L^*}{\partial \gamma}\Big\|\nonumber\\&\quad+\|e_\gamma[k]\|,\\
\|\tau[k+1]-\tau^*\|&\leq \Big\|\tau[k] - s\frac{\partial \overline{L}[k]}{\partial \tau} - \tau^* + s\frac{\partial L^*}{\partial \tau}\Big\|\nonumber\\&\quad+\|e_\tau[k]\|,
\end{align}
\end{subequations}
yielding \eqref{projection}. 
\end{proof}
Next, it follows from the mean value theorem \cite[Theorem~5.1]{Rudin} applied to each component in $\nabla f(g^{(p)}[k]) - \nabla f(g^{(p)*})$ that
\begin{align}
\nabla f(g^{(p)}[k]) - \nabla f(g^{(p)*}) = \nabla^2f(\upsilon[k])(g^{(p)}[k]-g^{(p)*}), \label{mean-value-th}
\end{align}
where $\upsilon[k]\coloneqq[\upsilon_1[k],\upsilon_2[k],\dots,\upsilon_n[k]]^\top$, with $\upsilon_i[k]$ lying on the line segment connecting $g^{(p)}_i[k]$ and $g^{(p)*}_i$, and $\nabla^2f(\upsilon[k])$ is the Hessian of $f(\mathrm{x})$ at $\mathrm{x}=\upsilon[k]$. Then, by using~\eqref{mean-value-th}, we have that
\begin{align}
G[k]-G^*&=F[k]z[k]+H[k]\begin{bmatrix}x^*\\\omega^*\\\gamma^*\\ \tau^*\end{bmatrix},
\label{ave_dyn2}
\end{align}
where
\begin{equation*}
H[k] \coloneqq \begin{bmatrix}\mathbf{0}&s(C^*-C[k])\\s(C[k]^\top-C^{*\top})&\mathbf{0} \end{bmatrix},
\end{equation*}
for some suitable matrices~$C^*$ and $C[k]$ corresponding to $(x^*,\omega^*,\gamma^*,\tau^*)$ and $(\overline{x}[k],\overline{\omega}[k],\overline{\gamma}[k],\tau[k])$, respectively, and
\begin{equation*}
F[k] \coloneqq \begin{bmatrix}I - s(D[k]+\Upsilon)&-sC[k]\\sC[k]^\top&I \end{bmatrix},
\end{equation*}
where $\Upsilon\in\mathds{R}^{(3n+4|\mathcal{E}_p|)\times(3n+4|\mathcal{E}_p|)}$ is a positive-semidefinite matrix, and the matrix~$D[k]$ is a block diagonal matrix given by
\begin{equation*}
D[k] \coloneqq \begin{bmatrix}\nabla^2f(\upsilon[k])& & & & \\ &\mathbf{0}_{n\times n} & & &\\ && \rho I_{n} && \\&&& \mathbf{0}_{3|\mathcal{E}_p|\times3|\mathcal{E}_p|} &\\&&&& \rho I_{|\mathcal{E}_p|}\end{bmatrix},
\end{equation*}
where $\mathbf{0}_{n\times n}\in\mathds{R}^{n\times n}$ is the all-zeros matrix, and $I_{n}\in\mathds{R}^{n\times n}$ is the identity matrix. We note that $F[k]$ is a skew-symmetric matrix resulting from multiplying $\frac{\partial L^{(i)}[k]}{\partial \tau_{li}}$ by a factor of $2$ in the~$\tau$-update in algorithm \eqref{pd_opf_un1}--\eqref{pd_opf_un2}. 
Define
\begin{equation*}
B[k] \coloneqq \begin{bmatrix}D[k] & C[k]\\-C[k]^\top & \mathbf{0}\end{bmatrix}
\end{equation*}
so that $F[k]=I-sB[k]$. 
The next result can be established by using some standard analysis.
\begin{lemma}\label{lem:B_rank}
If, at time instant~$k$,
\begin{align*}
B[k][\mathbf{0}_n^\mathrm{H},\mathrm{x}^\mathrm{H},\mathbf{0}_n^\mathrm{H},\mathrm{y}^\mathrm{H},\mathbf{0}_{|\mathcal{E}_p|}^\mathrm{H},\mathrm{z}^\mathrm{H}]^\mathrm{H}=0,
\end{align*}
for some $\mathrm{x}\in\mathds{R}^{n}$, $\mathrm{y}\in\mathds{R}^{3|\mathcal{E}_p|}$, and $\mathrm{z}\in\mathds{R}^{2+3|\mathcal{E}_p|}$, then, we have that $(\mathrm{x},\mathrm{y},\mathrm{z})=0$, where $\mathbf{0}_n$ denotes the all-zeros vector of length~$n$, and $\mathrm{x}^\mathrm{H}$ denotes the Hermitian transpose of $\mathrm{x}$.
\end{lemma}
Next, we show that all eigenvalues of $B[k]$ have a strictly positive real part. Suppose $\mu$ is an eigenvalue of $B[k]$ and $[\zeta^\mathrm{H},w^\mathrm{H}]^\mathrm{H}$ is an eigenvector corresponding to $\mu$, where $\zeta \coloneqq [\zeta^{(1)\top},\zeta^{(2)\top},\zeta^{(3)\top},\zeta^{(4)\top},\zeta^{(5)\top}]^\top$ has the same number of rows as $D[k]$ such that $\zeta^{(1)}\in\mathds{C}^{n}$, $\zeta^{(2)}\in\mathds{C}^{n}$, $\zeta^{(3)}\in\mathds{C}^{n}$, $\zeta^{(4)}\in\mathds{C}^{3|\mathcal{E}_p|}$, and $\zeta^{(5)}\in\mathds{C}^{|\mathcal{E}_p|}$. Then, on the one hand, we have that
\begin{align*}
\operatorname{Re}\left([\zeta^\mathrm{H},w^\mathrm{H}]B[k]\begin{bmatrix}\zeta\\ w \end{bmatrix}\right)&=\operatorname{Re}\left(\mu[\zeta^\mathrm{H},w^\mathrm{H}]\begin{bmatrix}\zeta\\ w \end{bmatrix}\right)\nonumber\\&=\operatorname{Re}(\mu)(\|\zeta\|_2^2+\|w\|_2^2).
\end{align*}
On the other hand, we have that
\begin{align*}
&\operatorname{Re}\left([\zeta^\mathrm{H},w^\mathrm{H}]B[k]\begin{bmatrix}\zeta\\ w \end{bmatrix}\right)=\operatorname{Re}\Big(\zeta^\mathrm{H}D[k]\zeta+\zeta^\mathrm{H}C[k]w\nonumber\\&\quad-w^\mathrm{H}C[k]^\top \zeta\Big)=\zeta^{(1)\mathrm{H}}\nabla^2f(\upsilon[k])\zeta^{(1)} + \rho\|\zeta^{(3)}\|_2^2 \nonumber\\&\quad+ \rho\|\zeta^{(5)}\|_2^2 + \zeta^{\mathrm{H}}\Upsilon\zeta>0,
\end{align*}
if $(\zeta^{(1)},\zeta^{(3)},\zeta^{(5)}) \neq 0$, or $\zeta\notin\mathrm{null}(\Upsilon)$, denoting the null space of $\Upsilon$. If $\operatorname{Re}(\mu)=0$, then, it follows that $\zeta\in\mathrm{null}(\Upsilon)$, $(\zeta^{(1)},\zeta^{(3)},\zeta^{(5)}) = 0$, and 
\begin{equation}
B[k][\mathbf{0}_n^{\mathrm{H}},\zeta^{(2)\mathrm{H}},\mathbf{0}_n^{\mathrm{H}},\zeta^{(4)\mathrm{H}},\mathbf{0}_{|\mathcal{E}_p|}^\mathrm{H},w^\mathrm{H}]^\mathrm{H}= 0.
%B[k]\begin{bmatrix}\mathbf{0}_n\\\zeta^{(2)}\\ \mathbf{0}_{|\mathcal{E}_p|}\\w \end{bmatrix} = 0,
\label{eq1}
\end{equation}
By Lemma~\ref{lem:B_rank}, \eqref{eq1} holds only if $\zeta = 0$ and $w=0$, which contradicts the fact that $[\zeta^\mathrm{H},w^\mathrm{H}]^\mathrm{H} \neq 0$. Therefore, all eigenvalues of $B[k]$ have a strictly positive real part, and, for small enough~$s$, the spectral radius of $F[k]$ denoted by $\rho(F[k])$ is strictly less than $1$.
%By using some standard analysis, it can be shown that there exists an induced matrix norm $\|\cdot\|$ such that $\|F[k]\|\leq b$, $\forall k$ and some $b<1$.
The next results can be established by using some standard analysis.
\begin{lemma}\label{lem:F_norm}
For sufficiently small~$s$, there exists an induced matrix norm~$\|\cdot\|$ such that $\|F[k]\|\leq b_1$, for some $b_1<1$, $\forall k$.
\end{lemma}
\begin{lemma}\label{lem:H_norm}
\begin{align*}
\|H[k][x^{*\mathrm{H}},\omega^{*\mathrm{H}},\gamma^{*\mathrm{H}},\tau^{*\mathrm{H}}]^\mathrm{H}\|\leq sb_2\|z[k]\|,
\end{align*}
for some positive~$b_2$, $\forall k$.
\end{lemma}

Taking $\|\cdot\|$ on both sides of \eqref{ave_dyn2} and applying the triangle inequality yields
\begin{align}
\|G[k]-G^*\|&\leq\|F[k]\|\|z[k]\|\nonumber\\
&\quad+\|H[k][x^{*\mathrm{H}},\omega^{*\mathrm{H}},\gamma^{*\mathrm{H}},\tau^{*\mathrm{H}}]^\mathrm{H}\|.\label{H1_opf_un_ineq1}
\end{align}
By applying Lemmas~\ref{lem:projection}, \ref{lem:F_norm} and \ref{lem:H_norm} and using the inequality \eqref{H1_opf_un_ineq1}, we obtain that
\begin{align}
\|z[k+1]\|&\leq \|G[k]-G^*\| + \|e[k]\|\nonumber\\
&\leq\|F[k]\|\|z[k]\|+\|H[k][x^{*\mathrm{H}},\omega^{*\mathrm{H}},\gamma^{*\mathrm{H}},\tau^{*\mathrm{H}}]^\mathrm{H}\|\nonumber\\&\quad+\|e[k]\|\leq (b_1+sb_2)\|z[k]\|+\|e[k]\| \nonumber\\
&= b\|z[k]\|+\|e[k]\|,
\label{ave_dyn_ineq}
\end{align}
where $b\coloneqq b_1+sb_2<1$ for sufficiently small~$s$.
Now, by multiplying both sides of \eqref{ave_dyn_ineq} by $a^{-(k+1)}$, we obtain
\begin{align}
a^{-(k+1)}\|z[k+1]\|\leq \frac{b}{a} a^{-k}\|z[k]\|+a^{-(k+1)}\|e[k]\|.
\label{ave_dyn_ineq3}
\end{align}
Then, by taking $\max\limits_{0\leq k\leq K}(\cdot)$ on both sides of \eqref{ave_dyn_ineq3}, we obtain
\begin{align}\label{H1_opf_un_ineq2}
&\max\limits_{0\leq k\leq K}a^{-(k+1)}\|z[k+1]\|\leq \frac{b}{a} \max\limits_{0\leq k\leq K} a^{-k}\|z[k]\|\nonumber\\&\quad+\frac{1}{a}\max\limits_{0\leq k\leq K} a^{-k}\|e[k]\|
\nonumber\\&\leq \frac{b}{a} \max\limits_{0\leq k\leq K+1} a^{-k}\|z[k]\|+\frac{1}{a}\max\limits_{0\leq k\leq K+1} a^{-k}\|e[k]\|.
\end{align}
Since \[\max\limits_{0\leq k\leq K}a^{-(k+1)}\|z[k+1]\| = \max\limits_{0\leq k\leq K+1}a^{-k}\|z[k]\|-\|z[0]\|,\] 
the relation~\eqref{H1_opf_un_ineq2} can be written as
\begin{align}
\|z\|_{a,K+1} \leq \frac{b}{a}\|z\|_{a,K+1} + \frac{1}{a}\|e\|_{a,K+1}+\|z[0]\|,\label{z_to_e_un2}
\end{align}
where $\|z\|_{a,K}\coloneqq \max\limits_{0\leq k\leq K}a^{-k}\|z[k]\|$.
Since $\|z\|_{a,K+1}\geq\|z\|_{a,K}$, it follows from~\eqref{z_to_e_un2} that
\begin{align}
\|z\|_{a,K} \leq \frac{b}{a}\|z\|_{a,K+1} + \frac{1}{a}\|e\|_{a,K+1}+\|z[0]\|,\label{z_to_e_un3}
\end{align}
Then, after rearranging~\eqref{z_to_e_un3}, we obtain
\begin{align*}
|z\|_{a,K+1} \leq \frac{1}{a-b}\|e\|_{a,K+1}+\frac{a}{a-b}\|z[0]\|.
\end{align*}
%and \[\left(1-\frac{b}{a}\right)\|z\|_{a,K}\leq \frac{1}{a}\|e\|_{a,K}+\left(1-\frac{b}{a}\right)\|z[0]\|.\] 
Because $\|\cdot\|_2\leq \alpha\|\cdot\|$ and $\|\cdot\|\leq \beta\|\cdot\|_2$ for some $\alpha$ and $\beta$, we have that $\|z\|_{a,K}\geq {\|z\|_2^{a,K}}/{\alpha}$, $\|e\|_{a,K} \leq \beta\|e\|_2^{a,K}$. Hence,
\[\frac{1}{\alpha}\|z\|_2^{a,K}\leq \frac{\beta}{a-b}\|e\|_2^{a,K}+\frac{a}{a-b}\|z[0]\|,\]
which can be rewritten as 
\[
\|z\|_2^{a,K}\leq \alpha_1\|e\|_2^{a,K}+\beta_1,
\]
where \[\alpha_1 =  \frac{\beta\alpha}{a-b},\] and \[\beta_1 = \frac{a\alpha}{a-b}\|z[0]\|,\]
yielding~\eqref{z_to_e_opf_un}.
\end{proof}
Next, we establish that $\mathcal{H}^{opf}_2$ is finite-gain stable.
%%%%%%%%%%%%%%%%%%%%%%%%%%%
%%%%%%%% H2_opf_un is finite-gain stable %%%%%%%%
\begin{proposition}\label{prop:H2_opf_un}
Let Assumptions~\ref{objective_assumption} and \ref{assume_comm_model_un} hold.
Then, under~\eqref{H2_opf_un}, we have that
\begin{align}
\textnormal{\textbf{R2. }}\|e\|_2^{a,K}\leq s\alpha_2\|z\|_2^{a,K}+\beta_2.\label{e_to_z_opf_un}
\end{align}
for some positive~$\alpha_2$ and $\beta_2$, $a \in (0,1)$, and sufficiently small~$s>0$.
\end{proposition}
\begin{proof}
In our analysis, we need the following results stated without proofs.
\begin{lemma}\label{lem:e_to_zy}
\begin{align}
\|e\|_2^{a,K} \leq s\alpha_3\|z\|_2^{a,K} + \alpha_4\|\tilde{z}\|_2^{a,K}+s\alpha_5\|\tilde{y}\|_2^{a,K},%+ \alpha_0,
\label{e_to_zy}
\end{align}
for some constant~$\alpha_3$, $\alpha_4$, and $\alpha_5$,
where $\tilde{z}[k]\coloneqq \big[(\overline{p}[k] - \hat{p}[k])^\top, (\overline{q}[k] - \hat{q}[k])^\top,\allowbreak (\overline{\varepsilon}[k] - \hat{\varepsilon}[k])^\top, (\overline{\ell}[k] - \hat{\ell}[k])^\top, (\overline{\nu}[k] - \hat{\nu}[k])^\top, (\overline{\eta}[k] - \hat{\eta}[k])^\top\big]^\top$, $\tilde{y}[k]\coloneqq \big[(\overline{y}^{(p)}[k] - \hat{y}^{(p)}[k])^\top, (\overline{y}^{(q)}[k] - \hat{y}^{(q)}[k])^\top, (\overline{y}^{(\varepsilon)}[k] - \hat{y}^{(\varepsilon)}[k])^\top,(\overline{y}^{(\ell)}[k] - \hat{y}^{(\ell)}[k])^\top, (\overline{y}^{(\nu)}[k] - \hat{y}^{(\nu)}[k])^\top,(\overline{y}^{(\eta)}[k] - \hat{y}^{(\eta)}[k])^\top\big]^\top$.
\end{lemma}
\begin{lemma}\label{lem:z_to_y}
\begin{align}
\|\tilde{z}\|_2^{a,K} \leq s\beta_3\|\tilde{y}\|_2^{a,K}+\beta_0,
\label{zt_to_y}
\end{align}
for some constant~$\beta_3$ and $\beta_0$.
\end{lemma}
\begin{lemma}\label{lem:y_to_d}
For $w \in \{\ell, p, q, \varepsilon, \nu, \eta\}$, let 
\begin{align*}
\hat{\delta}^{(w)}_{ij}[k+1] & \coloneqq 2\left(\frac{\partial L^{(i)}[k+1]}{\partial \hat{w}_{ij}} - \frac{\partial L^{(i)}[k]}{\partial \hat{w}_{ij}}\right),(i,j)\in\mathcal{E}_p, \\
\breve{\delta}^{(w)}_{li}[k+1]&\coloneqq 2\left(\frac{\partial L^{(i)}[k+1]}{\partial \breve{w}_{li}} - \frac{\partial L^{(i)}[k]}{\partial \breve{w}_{li}}\right),(l,i)\in\mathcal{E},
\end{align*}
$\delta^{(w)}[k]\coloneqq [\hat{\delta}^{(w)}[k]^\top,\breve{\delta}^{(w)}[k]^\top]^\top$,
$\delta[k] \coloneqq [\delta^{(p)}[k]^\top,\allowbreak\delta^{(q)}[k]^\top,\delta^{(\varepsilon)}[k]^\top,\delta^{(\ell)}[k]^\top,\delta^{(\nu)}[k]^\top,\delta^{(\eta)}[k]^\top]^\top$.
Then, for some constant~$\zeta_0$ and $\zeta_1$, the following relation holds:
\begin{align}
\|\tilde{y}\|_2^{a,K} \leq \zeta_1\|\delta\|_2^{a,K}+\zeta_0.\label{y_to_d}
\end{align}
\end{lemma}
\begin{lemma}\label{lem:d_to_zzt}
\begin{align}
\|\delta\|_2^{a,K} \leq \kappa_1\|z\|_2^{a,K}+\kappa_0,
\label{d_to_zzt}
\end{align}
for some constant~$\kappa_0$ and $\kappa_1$.
\end{lemma}
By substituting \eqref{zt_to_y} for $\|\tilde{z}\|_2^{a,K}$ in \eqref{e_to_zy}, we obtain that
\begin{align}
\|e\|_2^{a,K}& \leq s\alpha_3\|z\|_2^{a,K} + \alpha_4\|\tilde{z}\|_2^{a,K}+s\alpha_5\|\tilde{y}\|_2^{a,K}\nonumber\\
&\leq s\alpha_3\|z\|_2^{a,K} + \alpha_4(s\beta_3\|\tilde{y}\|_2^{a,K}+\beta_0)+s\alpha_5\|\tilde{y}\|_2^{a,K}\nonumber\\
&= s(\alpha_4\beta_3+\alpha_5)\|\tilde{y}\|_2^{a,K} + s\alpha_3\|z\|_2^{a,K} + \alpha_4\beta_0\nonumber\\
&\leq s(\alpha_4\beta_3+\alpha_5)(\zeta_1\|\delta\|_2^{a,K}+\zeta_0) + s\alpha_3\|z\|_2^{a,K}\nonumber\\
&\quad + \alpha_4\beta_0= s\zeta_2\|\delta\|_2^{a,K}+s\alpha_3\|z\|_2^{a,K}+\zeta_3,
\label{e_to_d}
\end{align}
where in the last inequality we applied \eqref{y_to_d}, $\zeta_2\coloneqq (\alpha_4\beta_3+\alpha_5)\zeta_1$, and $\zeta_3\coloneqq s(\alpha_4\beta_3+\alpha_5)\zeta_0+ \alpha_4\beta_0$.
By substituting \eqref{d_to_zzt} for $\|\delta\|_2^{a,K}$ in \eqref{e_to_d}, we obtain that
\begin{align}
\|e\|_2^{a,K}&\leq s\zeta_2\big(\kappa_1\|z\|_2^{a,K}+\kappa_0\big)+s\alpha_3\|z\|_2^{a,K}+s\zeta_3,
\end{align}
which can be rewritten as 
\[
\|e\|_2^{a,K}\leq s\alpha_2\|z\|_2^{a,K}+\beta_2,
\]
where $\alpha_2 =  \zeta_2\kappa_1+\alpha_3$, and $\beta_2 = s\zeta_2\kappa_0+s\zeta_3$,
yielding~\eqref{e_to_z_opf_un}.
\end{proof}
By using Propositions~\ref{prop:H1_opf_un}--\ref{prop:H2_opf_un}, we now show that \eqref{pd_opf_un1}--\eqref{pd_opf_un2} converges geometrically fast.
\begin{proposition}\label{prop:small_gain_opf_un}
Let Assumptions~\ref{objective_assumption} and \ref{assume_comm_model_un} hold.
Then, under algorithm~\eqref{pd_opf_un1}--\eqref{pd_opf_un2}, we have that
\begin{align}
\|z\|_2^{a,K}\leq \beta,
\label{z_relation_opf_un}
\end{align}
for some~$a \in (0,1)$, $\beta>0$, and sufficiently small~$s>0$.
In particular, $\mathbf{\hat{x}}\coloneqq (g^{(p)},g^{(q)},v,\hat{p},\hat{q},\hat{\varepsilon},\hat{\ell})$ and $\mathbf{\breve{x}}\coloneqq (g^{(p)},g^{(q)},v,\allowbreak\breve{p},\breve{q},\breve{\varepsilon},\breve{\ell})$ converge to $\mathbf{x^*}\coloneqq(g^{(p)*}, g^{(q)*}, v^*, p^*, q^*, \varepsilon^*,\ell^*)$ at a geometric rate~$\mathcal{O}(a^k)$.
\end{proposition}
\begin{proof}
By using Propositions~\ref{prop:H1_opf_un} and \ref{prop:H2_opf_un}, it follows that
\begin{align*}
\|z\|_2^{a,K}\leq \alpha_1\|e\|_2^{a,K}+\beta_1 \leq \alpha_1(s\alpha_2\|z\|_2^{a,K}+\beta_2)+\beta_1,
\end{align*}
which, after rearranging, results in
\begin{align*}
\|z\|_2^{a,K}\leq \frac{\alpha_1\beta_2+\beta_1}{1-s\alpha_1\alpha_2}\eqqcolon \beta,
\end{align*}
yielding \eqref{z_relation_opf_un}.
Hence, for sufficiently small~$s$, we have that $s\alpha_1\alpha_2<1$, which ensures that $\beta$ is finite. 
Next, we show that $\mathbf{\hat{x}}[k]$ and $\mathbf{\breve{x}}[k]$ converge to $\mathbf{x^*}$.
By substituting \eqref{d_to_zzt} for $\|\delta\|_2^{a,K}$ in \eqref{y_to_d}, we have that
\begin{align}
\|\tilde{y}\|_2^{a,K} &\leq \zeta_1\|\delta\|_2^{a,K}+\zeta_0\leq \zeta_1(\kappa_1\|z\|_2^{a,K}+\kappa_0)+\zeta_0\nonumber\\&= \kappa_3\|z\|_2^{a,K}+\kappa_4,
\label{y_to_z}
\end{align}
for some positive~$\kappa_3$ and $\kappa_4$.
Then, by using \eqref{y_to_z} and \eqref{z_relation_opf_un} in \eqref{zt_to_y}, we have that
\begin{align}
\|\tilde{z}\|_2^{a,K} &\leq s\beta_3\|\tilde{y}\|_2^{a,K}+\beta_0\leq s\beta_3(\kappa_3\|z\|_2^{a,K}+\kappa_4)+\beta_0\nonumber\\&\leq \kappa_5,
\label{zt_to_z}
\end{align}
for some~$\kappa_5>0$. 
Letting $\mathbf{\overline{x}}\coloneqq(\mathbf{\hat{x}}+\mathbf{\hat{x}})/2$, and using the triangle inequality and the inequalities \eqref{z_to_e_opf_un} and \eqref{zt_to_z}, we obtain that 
\begin{align}
\|\mathbf{\hat{x}}-\mathbf{x^*}\|_2^{a,K}&=\|\mathbf{\overline{x}} - \mathbf{x^*} + \mathbf{\hat{x}}-\mathbf{\overline{x}}\|_2^{a,K}\leq \|z\|_2^{a,K} + \|\tilde{z}\|_2^{a,K}\nonumber\\&\leq \beta + \kappa_5,
\end{align} 
from where it follows that $\mathbf{\hat{x}}[k]$ and $\mathbf{\breve{x}}[k]$ converge to $\mathbf{x^*}$.
\end{proof}
Finally, we prove that $\mathbf{x^*}$ is the solution of rSOCP. [The proof is similar to the proof of Lemma~\ref{lem:sced_sol}.]
\begin{lemma}\label{lem:opf_sol}
Consider $(x^*,\omega^*,\gamma^*,\tau^*)$, namely, the equilibrium of \eqref{centr_pd_opf}.
Then, $\mathbf{x^*}$ is the solution of rSOCP.
\end{lemma}

%\section{Distributed OPF Over Time-Varying Directed Communication Graphs}\label{sec:OPF_dd}
%In this section, we present a distributed algorithm for solving rSOCP in \eqref{rOPF} over time-varying directed communication graphs. 
%Since unidirectional communication model is more general than bidirectional communication model that was considered in Section~\ref{sec:opf},
%the algorithm can be considered a robust version of the distributed algorithm \eqref{pd_opf_un1}.
\subsection{Time-Varying Directed Communication Graphs}
The design of the robust extension of~\eqref{pd_opf_un1}--\eqref{pd_opf_un2} follows exactly the development procedure outlined in Section~\ref{subsec:sced_dd}.
%When $\mathcal{G}^{(c)}[k]$ is directed, the iterations in \eqref{pd_sced_un} fail to converge. 

In the robust extension, we let node~$i$ perform the same updates for $g^{(p)}_i$, $g^{(q)}_i$, $v_i$, $\lambda_i$, $\mu_i$, $\tau_{li}$, $(l,i)\in\mathcal{E}_p$ as the ones in \eqref{pd_opf_un1}--\eqref{pd_opf_un2}, i.e.,
\begin{subequations}\label{pd_opf_dd0}
\begin{align}
\chi[k+1] &= \Big[\chi[k] - s\frac{\partial L^{(i)}[k]}{\partial \chi}\Big]_{\chi^{\min}}^{\chi^{\max}},\chi \in \{g^{(p)}_i,g^{(q)}_i,v_i\},\\
\psi[k+1] &= \psi[k] + s\frac{\partial L^{(i)}[k]}{\partial \psi},\psi \in \{\lambda_i,\mu_i\},\\
\tau_{li}[k+1] &= \Big[\tau_{li}[k] + 2s\frac{\partial L^{(i)}[k]}{\partial \tau_{li}}\Big]_+,
\end{align}
\end{subequations}
The only difference between algorithm~\eqref{pd_opf_un1}--\eqref{pd_opf_un2} and its robust extension is how the averaging step is performed in the updates of the estimates of the shared quantities, namely, $\hat{\omega}[k] \coloneqq [\hat{p}[k]^\top,\hat{q}[k]^\top,\hat{\varepsilon}[k]^\top,\hat{\ell}[k]^\top]^\top$, and $\breve{\omega}[k] \coloneqq [\breve{p}[k]^\top,\breve{q}[k]^\top,\breve{\varepsilon}[k]^\top,\breve{\ell}[k]^\top]^\top$, and the gradient vectors, $\hat{y}^{(w)}[k]\coloneqq [\{\hat{y}_{ij}^{(w)}[k]\}_{(i,j)\in\mathcal{E}_p}]$, and $\breve{y}^{(w)}[k]\coloneqq [\{\breve{y}_{li}^{(w)}[k]\}_{(l,i)\in\mathcal{E}_p}]$, $w \in \{\ell, p, q,\varepsilon,\nu,\eta\}$. In the robust extension, we make use of the alternating averaging protocol to execute the averaging step; namely, for each $(i,j)\in\mathcal{E}_p$, node $i$ runs the following iterations:
\begin{subequations}\label{pd_opf_dd1}
\begin{align}
\hat{w}_{ij}[k+1] &= \Big[(1-a_{ij}[k])\hat{w}_{ij}[k]+ a_{ij}[k]\breve{w}_{ij}[k]\nonumber\\
&\quad  - s\hat{y}^{(w)}_{ij}[k]\Big]_{w_{ij}^{\min}}^{w_{ij}^{\max}},\\ 
\hat{d}_{ij}[k+1] &= (1-a_{ij}[k])\hat{d}_{ij}[k]+ a_{ij}[k]\breve{d}_{ij}[k] \nonumber\\&\quad + s\hat{y}^{(d)}_{ij}[k],\\
\hat{y}^{(\psi)}_{ij}[k+1] &= (1-a_{ij}[k]) \hat{y}^{(\psi)}_{ij}[k] + a_{ij}[k]\breve{y}^{(\psi)}_{ij}[k]\nonumber\\
&\quad + 2\left(\frac{\partial L^{(i)}[k+1]}{\partial \hat{w}_{ij}} - \frac{\partial L^{(i)}[k]}{\partial \hat{w}_{ij}}\right),
\end{align}
\end{subequations}
while node~$j$ executes the iterations given below: 
\begin{subequations}\label{pd_opf_dd1}
\begin{align}
\breve{w}_{ij}[k+1] &= \Big[(1-a_{ji}[k])\breve{r}^{(w)}_{ij}[k]+ a_{ji}[k]\hat{r}^{(w)}_{ij}[k]\nonumber\\
&\quad+\breve{w}_{ij}[k] - \breve{r}^{(w)}_{ij}[k]- s\hat{y}^{(w)}_{ij}[k]\Big]_{w_{ij}^{\min}}^{w_{ij}^{\max}},\\
\breve{d}_{ij}[k+1] &= (1-a_{ij}[k])\breve{r}^{(d)}_{ij}[k]+ a_{ij}[k]\hat{r}^{(d)}_{ij}[k]\nonumber\\&\quad +\breve{d}_{ij}[k] - \breve{r}^{(d)}_{ij}[k] + s\breve{y}^{(d)}_{ij}[k],\\
\breve{y}^{(\psi)}_{ij}[k+1] &= (1-a_{ij}[k]) \breve{\rho}^{(\psi)}_{ij}[k] + a_{ij}[k]\hat{\rho}^{(\psi)}_{ij}[k]\nonumber\\
&\quad+\breve{y}^{(\psi)}_{ij}[k] - \breve{\rho}^{(\psi)}_{ij}[k]\nonumber\\
&\quad + 2\left(\frac{\partial L^{(j)}[k+1]}{\partial \breve{\psi}_{ij}} - \frac{\partial L^{(j)}[k]}{\partial \breve{\psi}_{ij}}\right),
\end{align}
\end{subequations}
with $w\in \{\ell, p, q, \varepsilon\}$, $d \in \{\nu,\eta\}$, and $\psi \in \{\ell, p, q,\varepsilon,\nu,\eta\}$, where $a_{ij}[k]$, $a_{ji}[k]$, $\hat{r}^{(w)}_{ij}[k]$, $\breve{r}^{(w)}_{ij}[k]$, $\hat{\rho}^{(w)}_{ij}[k]$, and $\breve{\rho}^{(w)}_{ij}[k]$, $(i,j)\in\mathcal{E}_p$, $w \in \{\ell, p, q,\varepsilon,\nu,\eta\}$, are updated using the alternating averaging protocol provided below:
\begin{subequations}\label{pd_opf_dd3}
\begin{align*}
\mbox{node~$i$:}\\
\breve{\phi}_{ij}[k] &= \left\{\begin{array}{l l}\phi_{ji}[k] &\mbox{if } (j,i)\in\vec{\mathcal{E}}_c[k],\\ \breve{\phi}_{ij}[k-1] & \mbox{otherwise,}\end{array}\right.\\
\phi_{ij}[k] &= \left\{\begin{array}{l l}\neg\phi_{ij}[k-1] &\mbox{if } \breve{\phi}_{ij}[k]\neq \breve{\phi}_{ij}[k-1],\\ \phi_{ij}[k-1] & \mbox{otherwise,}\end{array}\right.\\
a_{ij}[k] &= \left\{\begin{array}{l l}0.5 &\mbox{if } (j,i)\in\vec{\mathcal{E}}_c[k], \breve{\phi}_{ij}[k]\neq \breve{\phi}_{ij}[k-1],\\ 0 & \mbox{otherwise,}\end{array}\right.\\
\mbox{node~$j$:}\\
\hat{\phi}_{ij}[k] &= \left\{\begin{array}{l l}\phi_{ij}[k]&\mbox{if } (i,j)\in\vec{\mathcal{E}}_c[k],\\ \hat{\phi}_{ij}[k-1] & \mbox{otherwise,}\end{array}\right.\\
\phi_{ji}[k] &= \left\{\begin{array}{l l}\neg\phi_{ji}[k-1] &\mbox{if } \hat{\phi}_{ij}[k]\neq \hat{\phi}_{ij}[k-1],\\ \phi_{ji}[k-1] & \mbox{otherwise,}\end{array}\right.\\
a_{ji}[k] &= \left\{\begin{array}{l l}0.5 &\mbox{if } (i,j)\in\vec{\mathcal{E}}_c[k], \hat{\phi}_{ij}[k]\neq \hat{\phi}_{ij}[k-1],\\ 0 & \mbox{otherwise,}\end{array}\right.\\
\hat{r}^{(w)}_{ij}[k] &= \left\{\begin{array}{l l}\hat{w}_{ij}[t_k] &\mbox{if } (i,j)\in\vec{\mathcal{E}}_c[k], \hat{\phi}_{ij}[k]\neq \hat{\phi}_{ij}[k-1],\\ 0 & \mbox{otherwise,}\end{array}\right.\\
\breve{r}^{(w)}_{ij}[k] &= \left\{\begin{array}{l l}\breve{w}_{ij}[t_k] &\mbox{if } (i,j)\in\vec{\mathcal{E}}_c[k], \breve{\phi}_{ij}[k]\neq \breve{\phi}_{ij}[k-1],\\ 0 & \mbox{otherwise,}\end{array}\right.\\
\hat{\rho}^{(w)}_{ij}[k] &= \left\{\begin{array}{l l}\hat{y}^{(w)}_{ij}[t_k] &\mbox{if } (i,j)\in\vec{\mathcal{E}}_c[k], \hat{\phi}_{ij}[k]\neq \hat{\phi}_{ij}[k-1],\\ 0 & \mbox{otherwise,}\end{array}\right.\\
\breve{\rho}^{(w)}_{ij}[k] &= \left\{\begin{array}{l l}\breve{y}^{(w)}_{li}[t_k] &\mbox{if } (i,j)\in\vec{\mathcal{E}}_c[k], \breve{\phi}_{ij}[k]\neq \breve{\phi}_{ij}[k-1],\\ 0 & \mbox{otherwise.}\end{array}\right.
\tag{\ref{pd_opf_dd3}}
\end{align*}
\end{subequations}
The following result can be easily established following the analysis from the proofs of Propositions~\ref{prop:H1_opf_un}--\ref{prop:small_gain_opf_un}.
\begin{proposition}\label{prop:small_gain_opf_dd}
Let Assumptions~\ref{objective_assumption} and \ref{assume_comm_model_dd} hold.
Then, under algorithm~\eqref{pd_opf_dd0}--\eqref{pd_opf_dd3}, we have that
\begin{align}
\|z\|_2^{a,K}\leq \beta,
\label{z_relation_opf_dd}
\end{align}
for some~$a \in (0,1)$, $\beta>0$, and sufficiently small~$s>0$.
In particular, $\mathbf{\hat{x}}\coloneqq (g^{(p)},g^{(q)},v,\hat{p},\hat{q},\hat{\varepsilon},\hat{\ell})$ and $\mathbf{\breve{x}}\coloneqq (g^{(p)},g^{(q)},v,\allowbreak\breve{p},\breve{q},\breve{\varepsilon},\breve{\ell})$ converge to $\mathbf{x^*}\coloneqq(g^{(p)*}, g^{(q)*}, v^*, p^*, q^*, \varepsilon^*,\ell^*)$ at a geometric rate~$\mathcal{O}(a^k)$.
\end{proposition}

\subsection{Numerical Simulations}
\label{subsec:simulations_opf}
In the following, we present numerical results to illustrate the performance of the robust distributed primal-dual algorithm~\eqref{pd_opf_dd0}--\eqref{pd_opf_dd3} over directed graph $\vec{\mathcal{G}}^{(c)}[k]$ using the IEEE~$69$--bus radial test system \cite{Matpower}. 

%using the IEEE~$39$--bus test system \cite{Matpower}. 
%\begin{figure}
%    \centering\vspace{3pt} 
%	\includegraphics[trim=3.3cm 0cm 0cm 0.1cm, clip=true, scale=0.28]{Figures/IEEE39.png} 
%	\vspace{-7pt} 
%    \caption{IEEE~$39$-bus test system.}
%    \vspace{-5pt} 
%    \label{fig:39-bus}
%\end{figure}
A subset of buses are designated to have a DER. For a DER at bus~$i$, we choose $f_i(p_i) = a_ip_i^2$, where $a_i>0$ is randomly selected. It is assumed that each communication link becomes inactive with probability~$0.4$. The algorithm uses a constant stepsize~$s = 3\times10^{-2}$. For initialization, we use $v_i[0]=1$, $i \in \mathcal{V}_p$, $g^{(p)}[0] = l_p[0]$, $g^{(q)}[0] = l_q[0]$, and the initial values of the remaining variables, except for the gradients~$\hat{y}^{(w)}[0]$ and $\breve{y}^{(w)}[0]$, $w \in \{\ell, p, q,\varepsilon,\nu,\eta\}$, are set to zero. The initial values of the gradients are computed using \eqref{init}, except for $\hat{y}^{(\eta)}[0]$ and $\breve{y}^{(\eta)}[0]$, which are computed by neglecting the voltages, $v_i[0]$'s, i.e.,
$\hat{y}^{(\eta)}[0] = \hat{\varepsilon}[0]$, and $\breve{y}^{(\eta)}[0] = \breve{\varepsilon}[0]$ (following the suggestion in the discussion after \eqref{y_eta}).

%In the distributed implementation, communicating data takes much longer than one iteration executed by a computing device. Rather than the total number of iterations, the number of communication attempts can serve as a more appropriate performance metric to evaluate the practical usefulness of the algorithm.
%We believe that it is reasonable to assume that a computing device is able to perform a number of iterations (less than 100) between consecutive communication attempts. Let $m$ denote the number of iterations between consecutive communication attempts. In the numerical example, we used different values of $m$. We note that making $m$ large or even finding a minimum of the local Lagrangians, $L^{(i)}(\mathbf{x}^{(i)})$, $i\in\mathcal{V}_p$, does not necessarily make the performance better. On the contrary, keeping $m$ relatively small ($m<20$) often achieves a much better performance.

In Fig.~\ref{fig:num_results_69bus-admm}, we compare the performance of algorithm~\eqref{pd_opf_dd0}--\eqref{pd_opf_dd3} against that of the asynchronous ADMM proposed in \cite{GuHu17}. The asynchronous ADMM has two parameters~$\rho$ and $\alpha$. 
%$\rho$ and $\alpha$ must be sufficiently large in order for the asynchronous ADMM to converge to the optimal solution. Larger $\rho$ allows the asynchronous ADMM to stabilize around the optimal solution more quickly, but it takes longer to converge; having larger $\alpha$ also slows down the convergence \cite{GuHu17}.
Figure~\ref{fig:num_results_69bus-admm} shows the trajectory of the relative cost error of the obtained solutions, namely, \[\frac{|f(g^{(p)}[k])-f(g^{(p)*})|}{f(g^{(p)*})},\] and the evolution of the largest constraint violation for different values of $\rho$ and $\alpha$. In algorithm~\eqref{pd_opf_dd0}--\eqref{pd_opf_dd3}, each node runs $m=10$ iterations, where we recall that $m$ is the number of iterations between consecutive communication attempts. In ADMM, each node solves a local optimization problem ($x$- or $z$-update) at each iteration (between consecutive communication attempts). In general, the closed-form solutions of the local problems are not available, but, the local problems arising in the OPF problem \eqref{rOPF} admit the closed-form solutions as shown in \cite{PeLo18}. The results in Fig.~\ref{fig:num_results_69bus-admm} demonstrate that algorithm~\eqref{pd_opf_dd0}--\eqref{pd_opf_dd3} has geometric convergence rate, and might converge faster than the asynchronous ADMM since the latter has asymptotic convergence rate.
\begin{figure}
    \begin{subfigure}[t]{0.5\textwidth}
        \centering
        \includegraphics[trim=0cm 0cm 0cm 0cm, clip=true, scale=0.4]{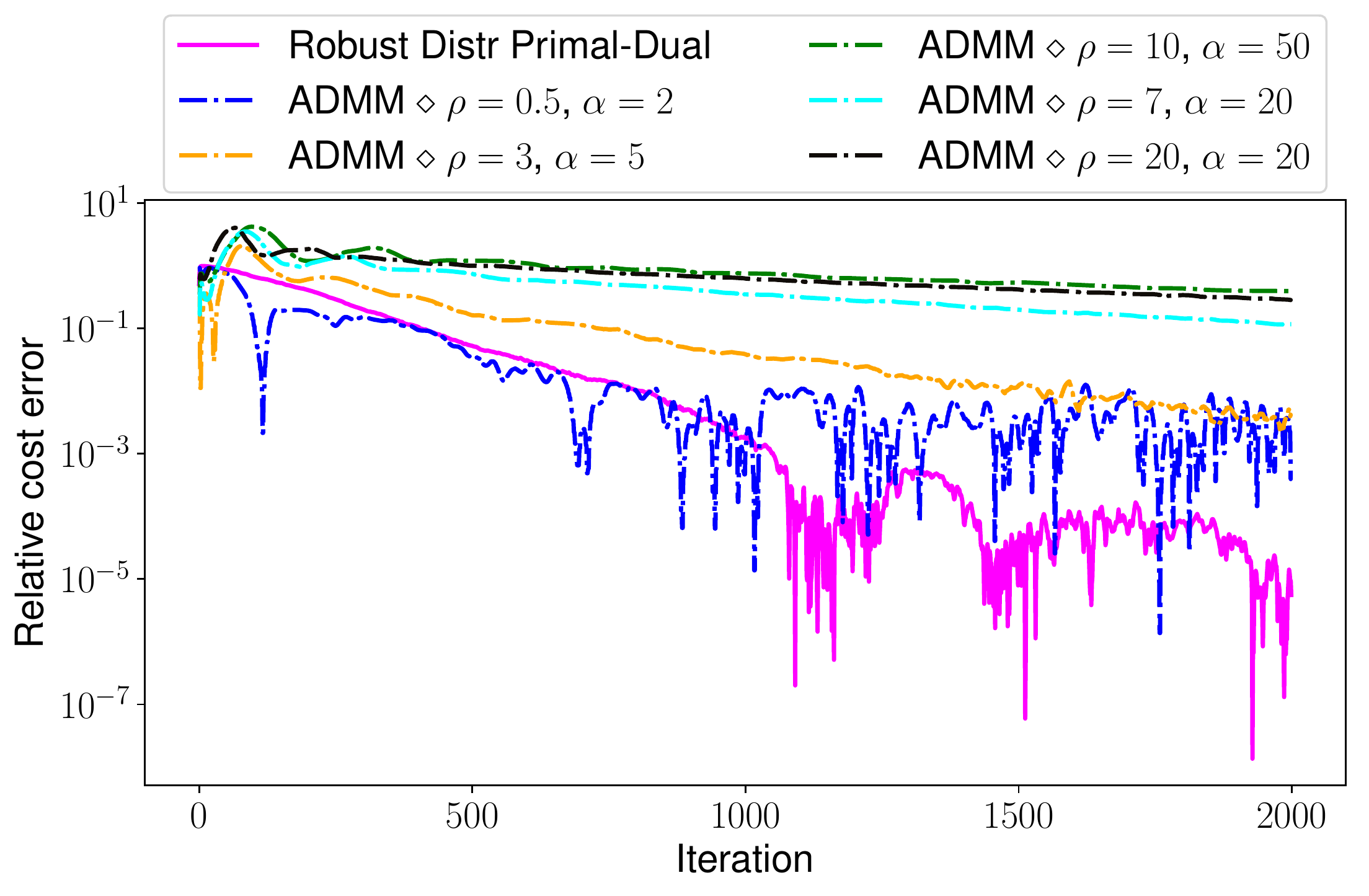}
        \vspace{-15pt}
        \caption{Trajectories of the relative cost error for the estimated solutions}
        \label{fig:cost-69-bus-admm}
    \end{subfigure}\vspace{10pt}\\
    \begin{subfigure}[t]{0.5\textwidth}
        \centering
        \includegraphics[trim=0cm 0cm 0cm 0cm, clip=true, scale=0.4]{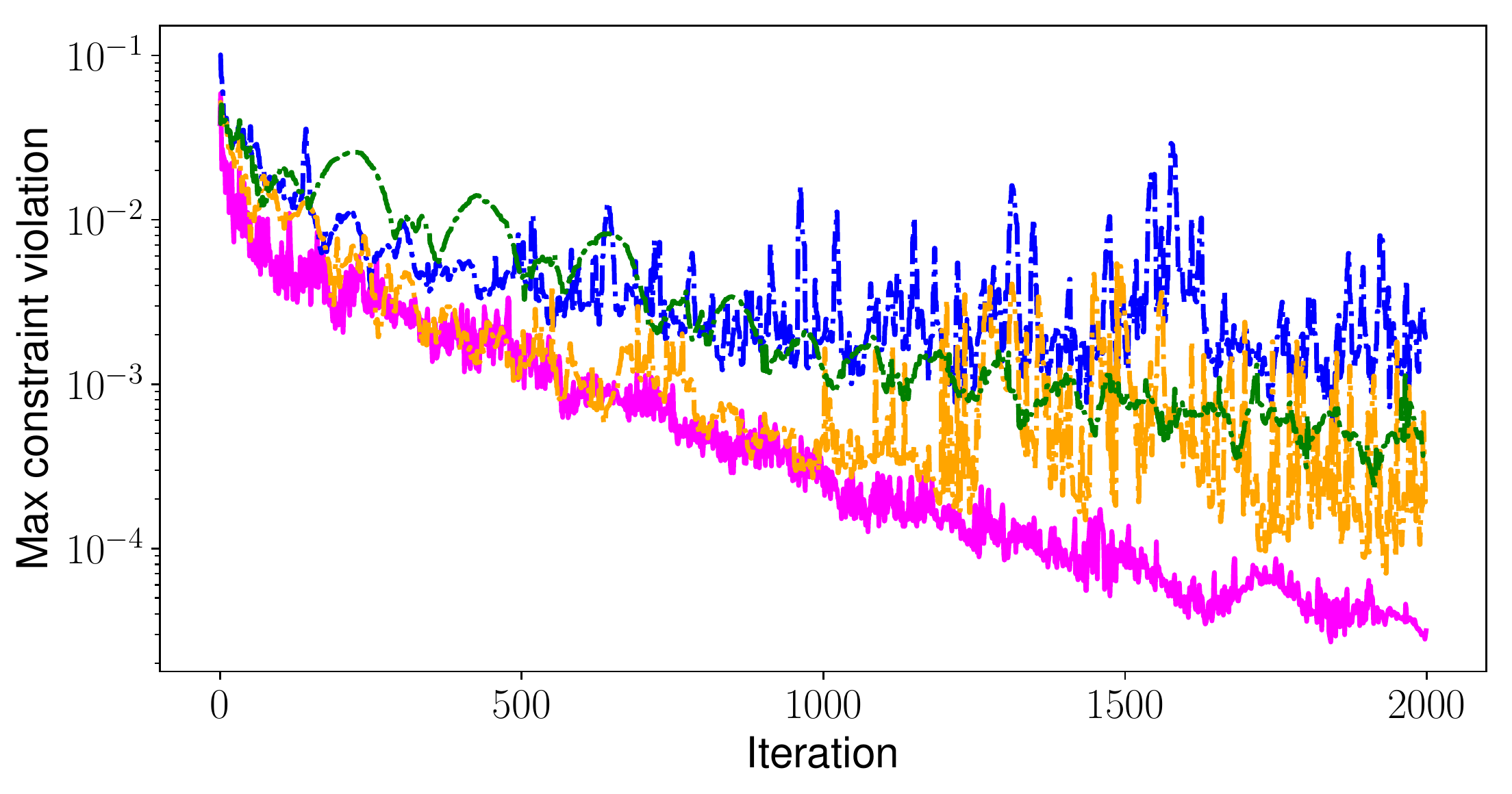}
        \vspace{-15pt}
        \caption{Evolution of the largest constraint violation -- part i}
        \label{fig:constr-69-bus-admm-1}
    \end{subfigure}\vspace{10pt}\\
        \begin{subfigure}[t]{0.5\textwidth}
        \centering
        \includegraphics[trim=0cm 0cm 0cm 0cm, clip=true, scale=0.4]{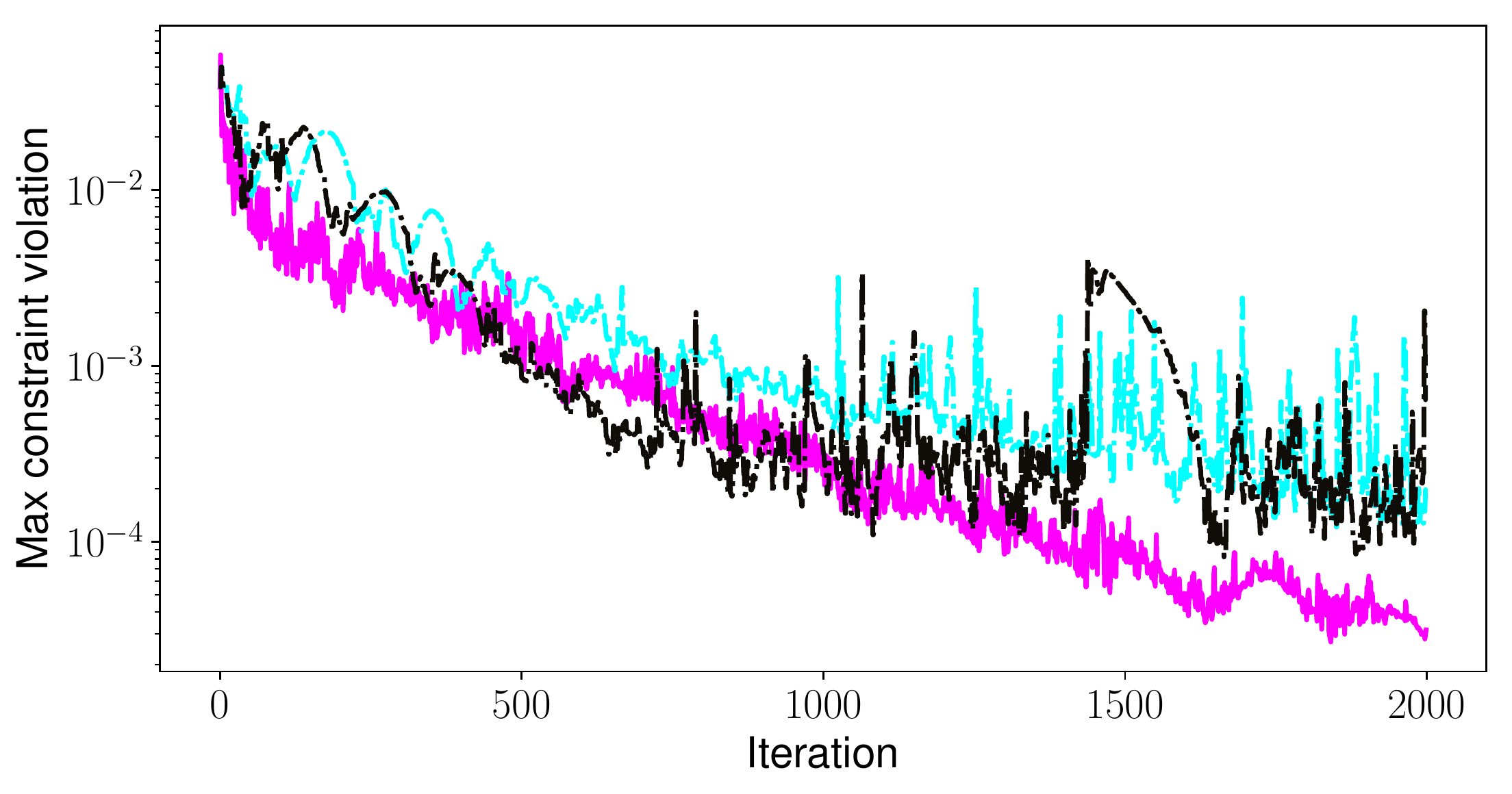}
        \vspace{-15pt}
        \caption{Evolution of the largest constraint violation -- part ii}
        \label{fig:constr-69-bus-admm-1}
    \end{subfigure}
    \vspace{10pt}
    \caption{A comparison of the performances of algorithm~\eqref{pd_opf_dd0}--\eqref{pd_opf_dd3} and the asynchronous ADMM.}
    \label{fig:num_results_69bus-admm}
\end{figure}

\section{Conclusion}
\label{sec:conclusion}
We presented distributed algorithms for solving the OPF problem for radial distribution systems over time-varying communication networks. The algorithms have geometric convergence rate and resiliency to communication delays and random data packet losses. One interesting future direction is to extend the proposed algorithms to solve multi-period OPF problems with battery energy storage systems.

\bibliographystyle{IEEEtran}
\bibliography{References}

\begin{thebibliography}{10}
\providecommand{\url}[1]{#1}
\csname url@samestyle\endcsname
\providecommand{\newblock}{\relax}
\providecommand{\bibinfo}[2]{#2}
\providecommand{\BIBentrySTDinterwordspacing}{\spaceskip=0pt\relax}
\providecommand{\BIBentryALTinterwordstretchfactor}{4}
\providecommand{\BIBentryALTinterwordspacing}{\spaceskip=\fontdimen2\font plus
\BIBentryALTinterwordstretchfactor\fontdimen3\font minus
  \fontdimen4\font\relax}
\providecommand{\BIBforeignlanguage}[2]{{%
\expandafter\ifx\csname l@#1\endcsname\relax
\typeout{** WARNING: IEEEtran.bst: No hyphenation pattern has been}%
\typeout{** loaded for the language `#1'. Using the pattern for}%
\typeout{** the default language instead.}%
\else
\language=\csname l@#1\endcsname
\fi
#2}}
\providecommand{\BIBdecl}{\relax}
\BIBdecl

\bibitem{BaWe08}
X.~Bai, H.~Wei, K.~Fujisawa, and Y.~Wang, ``Semidefinite programming for
  optimal power flow problems,'' \emph{International Journal of Electrical
  Power \& Energy Systems}, vol.~30, no.~6, pp. 383 -- 392, 2008.

\bibitem{FaLo13}
M.~{Farivar} and S.~H. {Low}, ``Branch flow model: Relaxations and
  convexification---part i,'' \emph{IEEE Transactions on Power Systems},
  vol.~28, no.~3, pp. 2554--2564, Aug. 2013.

\bibitem{GaLiLo15}
L.~{Gan}, N.~{Li}, U.~{Topcu}, and S.~H. {Low}, ``Exact convex relaxation of
  optimal power flow in radial networks,'' \emph{IEEE Transactions on Automatic
  Control}, vol.~60, no.~1, pp. 72--87, Jan. 2015.

\bibitem{DaZh13}
E.~{Dall'Anese}, H.~{Zhu}, and G.~B. {Giannakis}, ``Distributed optimal power
  flow for smart microgrids,'' \emph{IEEE Transactions on Smart Grid}, vol.~4,
  no.~3, pp. 1464--1475, Sep. 2013.

\bibitem{SuBa14}
P.~{{\v S}ulc}, S.~{Backhaus}, and M.~{Chertkov}, ``Optimal distributed control
  of reactive power via the alternating direction method of multipliers,''
  \emph{IEEE Trans. Energy Conversion}, vol.~29, no.~4, pp. 968--977, Dec.
  2014.

\bibitem{MaWe15}
S.~{Magn{\'u}sson}, P.~C. {Weeraddana}, and C.~{Fischione}, ``A distributed
  approach for the optimal power-flow problem based on admm and sequential
  convex approximations,'' \emph{IEEE Transactions on Control of Network
  Systems}, vol.~2, no.~3, pp. 238--253, Sep. 2015.

\bibitem{PeLo18}
Q.~{Peng} and S.~H. {Low}, ``Distributed optimal power flow algorithm for
  radial networks, i: Balanced single phase case,'' \emph{IEEE Transactions on
  Smart Grid}, vol.~9, no.~1, pp. 111--121, Jan. 2018.

\bibitem{GuHu17}
J.~{Guo}, G.~{Hug}, and O.~{Tonguz}, ``{Asynchronous ADMM for Distributed
  Non-Convex Optimization in Power Systems},'' \emph{arXiv e-prints}, p.
  arXiv:1710.08938, Oct 2017.

\bibitem{ZhCh12}
Z.~Zhang and M.~Y. Chow, ``Convergence analysis of the incremental cost
  consensus algorithm under different communication network topologies in a
  smart grid,'' \emph{IEEE Transactions on Power Systems}, vol.~27, no.~4, pp.
  1761--1768, Nov. 2012.

\bibitem{DoCaHa12}
A.~D. Dom{\'\i}nguez-Garc{\'\i}a, S.~T. Cady, and C.~N. Hadjicostis,
  ``Decentralized optimal dispatch of distributed energy resources,'' in
  \emph{Proc. IEEE Conf. Decision and Control}, Dec. 2012, pp. 3688--3693.

\bibitem{YaTa13}
S.~{Yang}, S.~{Tan}, and J.~{Xu}, ``Consensus based approach for economic
  dispatch problem in a smart grid,'' \emph{IEEE Transactions on Power
  Systems}, vol.~28, no.~4, pp. 4416--4426, Nov. 2013.

\bibitem{KaHu12}
S.~Kar and G.~Hug, ``Distributed robust economic dispatch in power systems: A
  consensus + innovations approach,'' in \emph{Proc. IEEE Power and Energy Soc.
  Gen. Meeting}, July 2012, pp. 1--8.

\bibitem{ZhPa14}
X.~Zhang and A.~Papachristodoulou, ``Redesigning generation control in power
  systems: Methodology, stability and delay robustness,'' in \emph{Proc. IEEE
  Conf. Decision and Control}, Dec. 2014, pp. 953--958.

\bibitem{CaDoHa15}
S.~T. Cady, A.~D. Dom\'{i}nguez-Garc\'{i}a, and C.~N. Hadjicostis, ``A
  distributed generation control architecture for islanded ac microgrids,''
  \emph{IEEE Transactions on Control Systems Technology}, vol.~23, no.~5, pp.
  1717--1735, Sept. 2015.

\bibitem{ChCo15}
A.~{Cherukuri} and J.~{Cort{\'e}s}, ``Distributed generator coordination for
  initialization and anytime optimization in economic dispatch,'' \emph{IEEE
  Trans. Control Network Syst.}, vol.~2, no.~3, pp. 226--237, Sep. 2015.

\bibitem{WuJo17}
J.~Wu, T.~{Yang}, D.~{Wu}, K.~{Kalsi}, and K.~H. {Johansson}, ``Distributed
  optimal dispatch of distributed energy resources over lossy communication
  networks,'' \emph{IEEE Transactions on Smart Grid}, vol.~8, no.~6, pp.
  3125--3137, Nov. 2017.

\bibitem{MoHu14}
J.~Mohammadi, G.~Hug, and S.~Kar, ``Role of communication on the convergence
  rate of fully distributed dc optimal power flow,'' in \emph{IEEE Int. Conf.
  on Smart Grid Commun.}, Nov. 2014, pp. 43--48.

\bibitem{MoHu15}
------, ``Fully distributed dc-opf approach for power flow control,'' in
  \emph{Proc. IEEE Power Energy Society General Meeting}, July 2015, pp. 1--5.

\bibitem{ZhFo19}
M.~Zholbaryssov, D.~Fooladivanda, and A.~D. Dom{\'\i}nguez-Garc{\'\i}a,
  ``Resilient distributed optimal generation dispatch for lossy ac
  microgrids,'' \emph{Systems \& Control Letters}, vol. 123, pp. 47 -- 54,
  2019.

\bibitem{Khalil}
H.~Khalil, \emph{Nonlinear Systems}, 3rd~ed.\hskip 1em plus 0.5em minus
  0.4em\relax Upper Saddle River, N.J.: Prentice Hall, 2002.

\bibitem{Nedic17}
A.~Nedi{\'c}, A.~Olshevsky, and W.~Shi, ``\BIBforeignlanguage{English
  (US)}{Achieving geometric convergence for distributed optimization over
  time-varying graphs},'' \emph{\BIBforeignlanguage{English (US)}{SIAM Journal
  on Optimization}}, vol.~27, no.~4, pp. 2597--2633, 2017.

\bibitem{BaWu89}
M.~{Baran} and F.~F. {Wu}, ``Optimal sizing of capacitors placed on a radial
  distribution system,'' \emph{IEEE Transactions on Power Delivery}, vol.~4,
  no.~1, pp. 735--743, Jan. 1989.

\bibitem{Nedic09}
A.~Nedi\'{c} and A.~Ozdaglar, ``Distributed subgradient methods for multi-agent
  optimization,'' \emph{IEEE Transactions on Automatic Control}, vol.~54,
  no.~1, pp. 48--61, Jan. 2009.

\bibitem{NonlinearProgramming}
D.~P. Bertsekas, \emph{Nonlinear Programming}, 2nd~ed.\hskip 1em plus 0.5em
  minus 0.4em\relax Athena Scientific, 1999.

\bibitem{ZhDo18}
M.~Zholbaryssov and A.~D. Dom\'{i}nguez-Garc\'{i}a, ``Microgrid distributed
  frequency control over time-varying communication networks,'' in \emph{Proc.
  IEEE Conf. Decision and Control}, Dec. 2018, pp. 5722--5727.

\bibitem{ZhDo19}
------, ``Convex relaxations of the network flow problem under cycle
  constraints,'' \emph{IEEE Trans. Control Network Syst.}, to appear.

\bibitem{Rudin}
W.~Rudin, \emph{Principles of Mathematical Analysis}, 3rd~ed.\hskip 1em plus
  0.5em minus 0.4em\relax McGraw-Hill, 1976.

\bibitem{Matpower}
R.~D. Zimmerman, C.~E. Murillo-Sanchez, and R.~J. Thomas, ``Matpower:
  Steady-state operations, planning, and analysis tools for power systems
  research and education,'' \emph{IEEE Transactions on Power Systems}, vol.~26,
  no.~1, pp. 12--19, Feb. 2011.

\end{thebibliography}
\end{document}